\documentclass[letterpaper,11pt,reqno]{amsart}

\makeatletter
\usepackage{silence}
\usepackage{amssymb}
\usepackage{latexsym}
\usepackage{amsbsy}
\usepackage{amsfonts}
\usepackage{hyperref}
\usepackage{graphicx}
\usepackage{enumerate}
\usepackage{mathtools}
\usepackage{color}
\usepackage{mathrsfs}
\usepackage[numbers,sort]{natbib}
\usepackage{listings}
\usepackage{xcolor} 

\usepackage{tikz}

\def\marginpar#1{\ignorespaces}

\newtheorem{theorem}{Theorem}[section]
\newtheorem{lemma}[theorem]{Lemma}
\newtheorem{proposition}[theorem]{Proposition}

\newtheorem{definition}[theorem]{Definition}

\newtheorem{remark}[theorem]{Remark}
\newtheorem{assumption}[theorem]{Assumption}

\numberwithin{equation}{section}
\makeatother
\begin{document}
\title[Order Flow Auction under Proposer-Builder Separation]{Analysis of the Order Flow Auction under Proposer-Builder Separation on Blockchain}

\author[Ruofei Ma, Wenpin Tang, David Yao]{{Ruofei} Ma, {Wenpin} Tang, {David} Yao}
\address{Department of Industrial Engineering and Operations Research, Columbia University. 
} \email{rm3881@columbia.edu,wt2319@columbia.edu,yao@columbia.edu}

\date{} 
\begin{abstract}
We study the impact of the order flow auction (OFA) in the context of the proposer-builder separation (PBS) mechanism in blockchains through a game-theoretic perspective. 
The OFA is designed to improve user welfare by redistributing maximal extractable value (MEV) to the users, in which two sequential auctions take place: the order flow auction and the block-building auction. 
We formulate the OFA as a multiplayer game, 
and establish the existence of a Nash equilibrium, and in the two-player case derive a closed-form solution (and prove its uniqueness) via a quartic equation.
Our result shows that the builder with a competitive advantage pays a lower cost, leading to a higher revenue, and adding to centralization in the builder space.
In contrast, the proposer's shares evolve as a martingale process, which implies decentralization in the proposer/validator space. 
Our analyses rely on various tools from stochastic processes, convex optimization, and polynomial equations. 
We also conduct numerical studies to corroborate our findings, and to bring out other features of the OFA under the PBS mechanism.
\end{abstract}

\maketitle



\section{Introduction}
\quad A blockchain is a decentralized, distributed, and tamper-proof digital ledger that tracks and verifies digital transactions securely without the need for a central authority. Its applications span a wide range of industries, including sustainable energies~\cite{wu2018energy}, cryptocurrency~\cite{hashemi2020cryptocurrency,LAROIYA2020213}, healthcare~\cite{ratta2021healthcare, cryptography3010003,bell2018applications}, and construction industry~\cite{KIM20202561}. 


\quad To maintain its decentralized structure, a blockchain relies on consensus mechanisms, with the two most widely used being Proof-of-Work (PoW) and Proof-of-Stake (PoS). PoW requires substantial computational power to solve cryptographic puzzles, making it highly energy-intensive. In contrast, PoS is more energy-efficient as it selects validators based on their stakes rather than computational power. 
In view of its energy efficiency and increasing adoption, we focus on the PoS mechanism in this study.

\quad The POS system, however, also faces its own challenges when large validators (or proposers -- here we use both terms interchangeably) gain enough power to negatively impact decentralization, the core principle of blockchain technology. 
This is because validators with larger holdings have a higher chance of validating new blocks and earn more rewards, increasing wealth concentration and excluding smaller participants~\cite{bains2022blockchain}. 
\smallskip
\setlength{\parindent}{0pt}\paragraph{\textbf{Maximal Extractable Value (MEV)}}
A significant portion of validators' revenue comes from MEV, which involves strategically reordering, inserting, or censoring transactions within a block to maximize profits. This includes activities like DEX arbitrage, sandwich attacks, and liquidations that exploit inefficiencies in the market. Some of these activities depend solely on the blockchain's state, using on-chain data to extract value. Others require information from external sources, such as off-chain data, to identify opportunities like CEX-DEX arbitrage~\cite{ethereum_mev, monoceros2023}.

\quad MEV is widely considered as one of the greatest threats to decentralization in blockchain networks, favoring the validators with more resources ~\cite{daian2020}. 
Extracting MEV effectively requires significant capital, advanced strategies, and considerable computational power, which most ordinary validators may not have access to. As a result, well-equipped validators can gain an advantage, further centralizing their control within the network~\cite{bahrani2024,capponi2024,ethereum_pbs,gupta2023,pai2023,yang2024}.
\smallskip
\setlength{\parindent}{0pt}\paragraph{\textbf{Proposer-Builder Separation (PBS)}}
To distribute MEV fairly among the validators and to prevent centralization, Ethereum has introduced the PBS mechanism ~\cite{ethereum_pbs}. Originally, validators were responsible for both proposing new blocks and constructing their contents.
PBS separates these responsibilities into two distinct roles: \textbf{block builders} and \textbf{block proposers}. 

\quad Block proposers (or validators) validate/propose blocks; whereas block builders are responsible for assembling blocks. Block builders compete to create the most profitable block, and participate in a block-building auction, offering fees (bids) to the proposer. The proposer then receives the block bundle created by the winning builder and their bid. 
This allows the proposers to collect significant auction revenue without requiring advanced technical expertise. By fostering competition among the builders, this system diminishes the advantage previously held by sophisticated validators over ordinary ones, facilitating a more balanced distribution of MEV among validators~\cite{ethereum_pbs, flashbots2024, yang2024, buterin2021}.

\quad While PBS can mitigate centralization among the validators to some extent, it tends to create centralization within the builder community. For instance, some builders are exceptionally skilled at exploiting arbitrage opportunities, enabling them to capture significant MEV and consistently win block-building opportunities. This concentration of expertise and resources among a few builders can lead to a situation where a small number of builders control a large portion of the block-building market, resulting in a concentration of power within the builder space~\citep{bahrani2024, gupta2023}.
\smallskip
\setlength{\parindent}{0pt}\paragraph{\textbf{Order Flow Auction (OFA)}}
OFA is another mechanism aiming to mitigate centralization and redistribute MEV. 
While PBS involves the builders and the validators, OFA focuses on the interaction between 
the \textbf{users} and the \textbf{builders}. 
It seeks to return a portion of MEV back to the users. In this process, users send their orders to a third-party auction, where block builders or MEV-extracting searchers bid for the exclusive rights to execute strategies on these orders. This approach is conceptually analogous to Payment For Order Flow (PFOF) in traditional finance.~\citet{gosselin2023}, and~\citet{monoceros2023} outline the structure of OFA as follows:
\begin{enumerate}
    \item Order Flow Originators (OFO): Order Flow Originators (OFOs) refer to wallets, decentralized applications (dApps), or custodians that users interact with for on-chain transactions. These OFOs gather the orders created by users and forward them to the OFA.
    \item Auctioneer: OFA discloses certain information to a group of bidders. 
    \item Bidders: Bidders then need to submit their bids back to the OFA. The OFA must determine the criteria for selecting the winning bids. 
    \item Winning Bid: Finally, the OFA bundles are sent to the winning bidder for inclusion in the blockchain. 
\end{enumerate}



\begin{figure}[h]
  \centering
   \begin{tikzpicture}[mynode/.style={draw, minimum width=1.5cm, minimum height=1cm, align=center}]
      \node[mynode][draw,rectangle] (Users) at (0,0) {Users};
      \node[mynode][draw,rectangle] (OFO) at (3,0) {OFO};
      \node[mynode][draw,rectangle]  (OFA) at (6,0) {OFA};
      \node[mynode][draw,rectangle]  (Builders) at (9,0) {Builders};
      \node[mynode][draw,rectangle]  (Validators) at (12,0) {Validators};

      \draw[->] (Users) -- (OFO) node[midway, above] {orders};
      \draw[->] (OFO) -- (OFA) node[midway, above] {orders};
      \draw[->] (Builders) -- (Validators) node[midway, above] {bids};
      \draw[->, bend left=15] (OFA) to (Builders) node[midway, above, xshift = 7.5cm, yshift = 0.3cm] {orders};
      \draw[->, bend left=15] (Builders) to (OFA) node[midway, below, xshift = 7.5cm, yshift = -0.3cm] {bids};
    \end{tikzpicture}
  \caption{OFA.}
  \label{fig:OFA} 
\end{figure}

\quad Figure ~\ref{fig:OFA} presents a flowchart illustrating the OFA process. It is important to note that an order included by a block builder is not necessarily guaranteed on-chain inclusion. The order will only be included if the block builder wins the block-building auction. Therefore, the OFA must provide reliable inclusion guarantees and faces the challenge of the double auction problem, which stems from the interaction between the OFA and the block-building auction~\cite{gosselin2023}. 
If higher rebates are offered to users, less MEV can be redistributed to validators, which may affect the timely inclusion of blocks. Thus, it is crucial to find the optimal balance between user rebates and inclusion fees offered to validators in the winner selection process~\cite{gosselin2023}.
\smallskip
\setlength{\parindent}{0pt}
\paragraph{\textbf{Interplay Between PBS and OFA}}
While PBS and OFA target different parts of the blockchain system, with PBS redistributing MEV between builders and validators and OFA seeking to return MEV to users, they are interconnected. In practice, these two mechanisms operate sequentially: builders first acquire users' order flow through OFA, then aggregate this order flow with other available transactions to construct a block, which they submit for inclusion via PBS.  This introduces a dependency: the value that a builder can extract from the order flow may influence their bid in PBS. Conversely, the likelihood of winning the block-building auction in PBS constrains the amount a builder is willing to pay for order flow in OFA. Hence, the two auctions form a coupled system in which builders optimize their total bid across both mechanisms. 

\smallskip
\setlength{\parindent}{0pt}\paragraph{\textbf{Main Contributions}} 
\begin{enumerate}
    \item We formulate the OFA as an $M$-player game, where each builder's decision variable is the amount of MEV they are willing to pay. Under suitable conditions, we establish the existence of a Nash equilibrium and provide its characterization.
    \item We study in depth the case when $M = 2$, i.e., there are two competing builders. Interestingly, solving for the corresponding Nash equilibrium boils down to solving a univariate quartic equation. We show that there exists a unique Nash equilibrium and derive its closed-form solution. Our analysis reveals that when there are two builders, their equilibrium payments do not scale linearly with the MEV they can extract. Instead, the more capable builder pays less, and hence   earns a higher expected revenue, which drives centralization. Our simulation experiments with three players show a similar pattern, with an interesting variation: the most capable player’s advantage is further amplified, while the gap between the second and the least capable players narrows compared to the two-player case. 
    \item We formulate the evolution of validators' stake shares as a Pólya urn process with a random replacement matrix. With this model, we show the stake shares follow a martingale process. This aligns with our simulation results, which demonstrate that the average stake shares remain nearly constant over time. Furthermore, we analyze the long-term behavior of this process. In the absence of consumption factors (i.e., when no costs are incurred due to staking), we characterize the distribution of the limiting stake shares via functional equations.
\end{enumerate}

\smallskip
\setlength{\parindent}{0pt}\paragraph{\textbf{Related Literature}}
Our paper relates closely to the literature on centralization in blockchains. Prior studies have explored various factors driving centralization, both with and without the PBS mechanism. 
Reward heterogeneity, a direct cause of centralization, arises from skill disparities among block producers in the absence of PBS~\cite{bahrani2024}. It is further influenced by order flow acquisition, along with the resulting MEV extraction and arbitrage opportunities under PBS~\cite{capponi2024,gupta2023}. In fact, the role of order flow in centralization has been widely recognized. Empirical studies show that private order flow exacerbates disparities in block-building capacity among builders~\cite{yang2024}, and those with greater access to private order flow have a higher probability of winning the block-buildling auction~\cite{wang2024privateorderflowsbuilder}.

\quad Our study takes into account both the builder heterogeneity in skills and knowledge, and the process of order flow acquisition. Relative to previous studies, we further model order flow acquisition as an auction that explicitly captures its relationship with builders' block-building capacity, while also incorporating the interaction between the order flow auction and the block-building auction, which appears under-explored in prior works. 

\quad Our research also contributes to the literature on the evolution of validators' stake shares in the Proof-of-Stake (PoS) system.~\citet{rosu2021market}, and~\citet{Tang22} show that in a standard PoS system, validators' stake shares evolve as a martingale process.~\citet{tang2024polynomialvotingrules} analyze stake share evolution under a polynomial voting rule. 
See also \citet{Tang24} for a review. 
Building on these works, our research further examines a setting in which the validator's reward for proposing a block is stochastic and dependent on builders' bids. In addition, we include consumption factors into the system, recognizing that staking in the pool may incur costs, including opportunity costs, locked funds, and other potential expenses.

\medskip
{\bf Organization of the paper}:
The remainder of the paper is organized as follows.
In Section \ref{sc2} we introduce the two models for OFA and for PBS. 
We analyze the Nash equilibrium among the builders in Section \ref{sc3}, and then
study the evolution of stake shares of the validators 
in Section \ref{sec: validators' shares}.  
Numerical examples and findings are presented in Section \ref{sec: numerical results}, and 
concluding remarks in Section \ref{concl}.

\section{The OFA and PBS Models}
\label{sc2}
\quad In this section, we develop a formal model for the OFA, PBS, and the PoS system, focusing on the equilibrium and stochastic processes associated with the model. 
Section~\ref{sec: modeling--builders} introduces the model for builders, who participate in both the order flow auction and the block-building auction. Section~\ref{sec: modeling--validators} presents the model for validators and the PoS system.

\quad First, here is a list of some of the common notations used throughout the paper.
\begin{itemize}
    \item $\mathbb{N}_+$ denotes the set of positive integers, $\mathbb{R}$ denotes the set of real numbers, and $\mathbb{R}_+$ denotes the set of positive real numbers.
    \item $\left[ n\right]$ denotes the set $\{1,2,\ldots,n \}$.
    \item $a = \mathcal{O}(b)$ means $\frac{a}{b}$ is bounded from above as $b \to \infty$, and  $a \sim b$ means $\frac{a}{b}$ converges to $1$ as $b \to \infty$.
    \item \( \mathbf{I} \left(A \right) \) denotes the indicator function of event \( A \), which equals $1$ if event $A$ happens and 0 otherwise. 
\end{itemize}

\quad Builders in a decentralized system are responsible for assembling transaction blocks. They participate in two sequential auctions: the order flow auction and the block-building auction. In the order flow auction, builders compete to acquire users' order flow, which can provide additional MEV opportunities. By strategically integrating the auctioned order with their existing order flow, they can optimize execution for increased profitability. In the subsequent block-building auction, builders bid for the right to propose their assembled block for inclusion on-chain, thereby capturing MEV generated from transaction sequencing and execution.

\quad Validators are responsible for validating and proposing blocks. They participate in the consensus process by staking cryptocurrencies as collateral. At each block-building opportunity, a validator is selected to propose the next block. The selected validator receives the bid from the winning builder in the block-building auction. 

\begin{remark}
{\em In our model, we do not distinguish between MEV searchers and builders, but instead treat them as a collective entity. 
This is motivated by the fact that searchers often direct their order flow to dominant builders or restrict it exclusively to vertically integrated builder-searcher entities \citep{gupta2023, titan2023builder, ethresear2024buildersbehavioralprofiles, yang2024}. This close coordination, coupled with their shared objective of maximizing MEV, has convinced us to ignore any distinction between their roles, and to focus instead on their combined strategic behavior in the context of MEV extraction.}
\end{remark}

\quad Let $M \in \mathbb{N}_+$ be the total number of builders, and $N \in \mathbb{N}_+$ be the total number of validators, both of which remain fixed throughout the paper; and let $\left[ M \right]$ and $\left[ N \right]$ denote the sets of all builders and all validators, respectively. Builders first participate in a single round of the order flow auction, followed by a single round of the block-building auction.


\subsection{Builder's Game} \label{sec: modeling--builders}
Let \( f_{i} \) represent the amount of Maximal Extractable Value (MEV) that builder \( i \) can capture. Assume that \( f_{i} \) follows a distribution \( D_{i} \) for each \( i \in \left[ M \right] \). The expected value of \( f_{i} \) is given by \( \mathbf{E} [f_{i}] = \bar{f}_{i} \). Suppose builder $i$ has value $v_{i}$ for winning the transaction right in the order flow auction. Assume that \( v_{i} \) follows some distribution \( F_{i} \), for \( i \in \left[ M \right] \). The expected value of \( v_{i} \) is given by \( \mathbf{E} [v_{i}] = \bar{v}_{i} \).

\quad  Let \( r_{{i}} \) denote the total amount that builder \( i \) bids to both users and the selected validator. $r_i$ is modeled as a random variable parametrized by $h_i$, such that $\mathbb{E} (r_i) = h_i$, with support $\left[r_i^{\min}, r_i^{\max} \right]$, where $r_i^{\min}> 0$. By definition, it follows that $h_i \in \left[r_i^{\min}, r_i^{\max} \right]$. Suppose a fraction \( \mu  r_{{i}} \) is ultimately bid to users, and \( (1-\mu) r_{i} \) is bid to the selected validator, with \( \mu \) determined by the design of the order flow auction $\left( 0 \leq \mu \leq 1\right)$. Note that the two auctions are sequential, and that builders must report, in the first auction—the order flow auction—the total amount they intend to bid across both the OFA and the PBS auctions. A fraction of this reported amount is bid to users through the OFA, while the remainder is bid to the selected validator through PBS \cite{gosselin2023}. We make the following assumption.




\begin{assumption}
The winners of the order flow auction and the block-building auction are determined independently, and both of them are first-price auctions. The probability that $r_i$ is the highest among $r_1, \dots, r_M$ is given by

\begin{equation}
    \mathbf{P} \left( r_{\left(M\right)} = r_i \right) = \frac{h_i}{\sum_{j=1}^M h_j}, \label{assump: PL model}
\end{equation}

where $r_{\left(M\right)} = \max \left\{ r_1, r_2,\dots,r_M\right\}$.

\end{assumption}

\begin{remark}
    {\em Equation \ref{assump: PL model} is based on the Plackett-Luce model from the theory of social choice. In this framework, each player $i$ is assigned a positive parameter $h_i$ representing their latent ``ability'' or ``competitiveness,'' which may reflect heterogeneity among builders, such as differences in resources or technical capabilities. The model implies that the probability of player $i$ submitting the highest bid $r_i$ among all players is proportional to $h_i$, normalized by the sum $\sum_{j=1}^M h_j$.}
\end{remark}



\quad Let $V_i$ denote the event that builder $i$ wins the order flow auction, and $Z_i$ denote the event that builder $i$ wins the block-building auction:
\begin{align}
\mathbf{I} \left( V_i\right) &= 
\begin{cases} 
1 & \text{if builder \( i \) wins the order flow auction,} \\ 
0 & \text{otherwise.} 
\end{cases} \notag\\
\mathbf{I} \left( Z_i\right) &= 
\begin{cases} 
1 & \text{if builder \( i \) wins the block-building auction,} \\ 
0 & \text{otherwise.} 
\end{cases} \notag
\end{align}

Table~\ref{tab: revenue outcomes for builders} presents the four possible outcomes of the two auctions for builder $i$, along with the corresponding revenue (utility) in each case. (Revenue can be regarded as utility, and we use both terms interchangeably in this paper.) Let $h_{-i}$ denote the $M - 1$ strategies of all the builders except $i$. The expected utility of builder \( i \) is given by:
\begin{align}
 \pi_i(h_i | h_{-i}) &= \mathbf{E} \left[ \mathbf{I} \left( Z_i \text{} \right) (f_i + v_i \mathbf{I} \left(V_i\right)) - \mu r_i \mathbf{I}  \left(V_i\right) - (1 - \mu) r_i \mathbf{I} \left(Z_i \right) \right], \notag\\
 &= \bar{f}_i \frac{h_i}{\sum_{j=1}^M h_j} + \bar{v}_i \left( \frac{h_i}{\sum_{j=1}^M h_j}\right)^2 - \frac{h_i^2}{\sum_{j=1}^M h_j}.
 \label{eq: utility func}
\end{align}

Each builder $i$ chooses a bid with expectation $h_i$ from the strategy space $B_i \coloneqq \left[r_i^{\min},r_i^{\max}\right]$, where $r_i^{\min} > 0$, to maximize their expected utility. Let $h_i^*$ denote the optimal bid, i.e.,


\begin{equation*}
h_{i}^* = \arg \max_{h_i \in B_i} \bar{f}_i \frac{h_i}{\sum_{j=1}^M h_j} + \bar{v}_i \left( \frac{h_i}{\sum_{j=1}^M h_j}\right)^2 - \frac{h_i^2}{\sum_{j=1}^M h_j}.
\end{equation*}

\quad In contrast to ~\citet{capponi2024}, which introduce an additional player, the order flow provider, and model the order flow acquisition process using a quadratic function, our framework takes a different perspective. Specifically, we model the payment as a proportion of the MEV obtained and introduce randomness in the winner selection process for both the OFA and PBS. As a result, our objective function involves the term $\frac{h_i}{\sum_{j=1}^M h_j}$, reflecting the competitive interaction among the builders.

\begin{table}[h]
    \centering
    \begin{tabular}{|c|c|c|}
        \hline
        Order Flow Auction & Block-Building Auction & Revenue \\ 
        \hline
        Win   & Win   & $f_i + v_i - r_i $  \\ 
        Win   & Lose   & $-\mu r_i$   \\ 
        Lose   & Win   & $f_i - \left( 1-\mu \right) r_i$   \\ 
        Lose   & Lose   & $0$ \\
        \hline 
    \end{tabular}

    \medskip
    \caption{Revenue Outcomes for Builder $i$.}
    \label{tab: revenue outcomes for builders}
\end{table}

\quad We aim to find the Nash equilibrium of the game among the builders.

\begin{definition}
Let $M$ denote the total number of builders. Let $B_i$ be the set of all possible strategies for builder $i$, where $i \in \left[ M \right]$. Let $h = (h_i, h_{-i})$ be a strategy profile where $h_{-i}$ denotes the $M - 1$ strategies of all the builders except $i$. A Nash equilibrium is a strategy profile $h^* = (h_i^*, h_{-i}^*)$ if
$$\pi_i(h_i^* \vert h_{-i}^*) \geq \pi_i(h_i \vert h_{-i}^*)$$ 
for all $h_i \in B_i$.
\end{definition}

\subsection{Validator's Game} \label{sec: modeling--validators}
Let $s_{j,t}$ denote the stake held by validator $j$ at time $t$, and define the total stake at time $t$ as $S_t \coloneqq \sum_{j = 1}^N s_{j,t}$. The fraction of the total stake held by validator $j$ at time $t$ is given by $\omega_{j,t} \coloneqq \frac{s_{j,t}}{S_t}$ for $j \in \left[ N \right]$. In each round $t$, the probability that validator $j$ is selected to propose a block is $\omega_{j,t-1}$. The initial stake share of validator $j$ is given by $\omega_{j,0} = \frac{s_{j,0}}{S_0}$, where $s_{j,0}$ represents the initial stake held by validator $j$, and $S_0 = \sum_{j=1}^N s_{j,0}$ denotes the total initial stake held by all $N$ validators. We assume that each validator holds a positive initial stake.
\begin{assumption}
    For all $j \in \left[ N\right]$, the initial stake is strictly positive, i.e., $s_{j,0} > 0$, and hence the total initial stake satisfies $S_0 = \sum_{j=1}^N s_{j,0} > 0$.
\end{assumption}


\quad If selected, a validator can either propose the block submitted by the winning builder in the block-building auction or choose to propose the block built by themselves. Let $\beta_{w}$ denote the bid submitted by the winning builder in the block-building auction. $\beta_{w}$ is a random variable drawn from the set of random variables $\left\{\mu r_1, \mu r_2, \dots, \mu r_M \right\}$, where $\beta_{w} =\mu r_i$ with probability $\frac{h_i}{\sum_{k=1}^M h_k}$.  Let $\beta_{v}$ denote the value of the block built directly by the selected validator. We make the following assumptions:


\begin{assumption}
    MEV extraction abilities are identical across all validators and are characterized by the same value $\beta_{v}$.
\end{assumption}


\begin{remark}
{\em Similar to the builders, we consider a validator's ability to extract MEV as an intrinsic characteristic. Moreover, 
the validators with significantly higher skills in MEV extraction 
would likely operate as the builders instead. Consequently, those in the role of the validators are expected to exhibit relatively homogeneous MEV extraction capabilities in the PBS system.}
\end{remark}

\quad The reward received by the selected validator for proposing a block is given by $R_t = \max \{\beta_{w}, \beta_{v}\} $.  $\left\{R_t\right\}_{t\geq1}$ is a sequence of i.i.d random variables with mean $\mathbb{E}\left( R_t\right) = R$ for all $t \geq 1$.

\quad Let $R_{\min} \coloneqq \min\left\{ k_1, k_2, \ldots, k_M \right\}$ denote the minimum possible value of the i.i.d variable $R_t$, where $k_i = \max\left\{ \beta_v, \mu r_i^{\min} \right\}$ for each $i \in \left[ M\right]$. Similarly, let $R_{\max} \coloneqq \max\left\{ g_1, g_2, \ldots, g_M \right\}$ denote the maximum possible value of $R_t$, where $g_i = \max\left\{ \beta_v, \mu r_i^{\max} \right\}$ for $i \in \left[ M\right]$.

\quad Additionally, a staking cost $ \alpha \frac{s_{j,t}}{S_t^{1+\gamma}}$ is incurred, where $\gamma \geq 0$. This cost may reflect factors such as opportunity cost of locked funds or operational expenses. It is proportional to a validator's stake share, conceptually similar to transaction fees in traditional finance, which are typically a small percentage of the transaction amount. To ensure that $\alpha$ is sufficiently small, we make the following assumption.


\begin{assumption}
     $\alpha < \min \left\{ S_0^\gamma R_{\min}, \, S_0, \, R_{\min} \right\}$. \label{assump: bound on alpha}
\end{assumption}

\quad This assumption guarantees that $S_t$ is strictly increasing for all $t \geq 0$ (see Lemma~\ref{lemma: S_t increasing}), and connects the cost of staking with the builder’s bid.


\quad Let $X_{j,t}$ denote the event that validator $j$ is chosen at time $t$, and define its corresponding indicator variable $\mathbf{1}\left(X_{j,t}\right)$ as
\begin{equation}
    \mathbf{1}\left(X_{j,t}\right) = \begin{cases}
        1, & \text{ if validator $j$ is selected at time $t$}, \\
        0, & \text{ otherwise.}
    \end{cases} \notag
\end{equation}
\begin{assumption} \label{assump: independence of validator selection}
The event $X_{j,t}$, i.e., the selection of validator $j$ at time $t$, is independent of all other random variables in the system.
\end{assumption}

\quad The stake held by each validator evolves according to the following update rule:
\begin{equation}
\begin{aligned}
    s_{j,t} &= s_{j, t-1} + R_t \mathbf{1}\left(X_{j,t}\right) - \alpha \frac{s_{j,t-1}}{S_{t-1}^{1+\gamma}}, \\
    &=\begin{cases} 
    s_{j,t-1} - \alpha \frac{s_{j,t-1}}{S_{t-1}^{1+\gamma}} & \text{with probability } 1 - \omega_{j,t-1} \\ 
        s_{j,t-1} + R_{t} - \alpha \frac{s_{j,t-1}}{S_{t-1}^{1+\gamma}}& \text{with probability } \omega_{j,t-1}
    \end{cases}
    \quad\text{ for } j \in \left[N\right].
\end{aligned}
\end{equation}
As a result, the total stake \( S_t \) evolves as  
\begin{equation}
    S_t = S_{t-1} + R_{t} -\frac{\alpha}{S_{t-1}^{\gamma}}
    \label{eq:S_t},
\end{equation}
where we take into account the identity \( \sum_{k=1}^{N} s_{k,t} = S_t \).
Finally, for each $t \in \mathbb{N_{+}}$, let $\mathcal{F}_t$ denote the filtration generated by random events $\left( X_{j,r}: j \in \left[N\right], r\leq t\right)$.

\section{Analysis of the game of the builders}
\label{sc3}

\quad In this section, we analyze the strategic interactions among the builders.
Section~\ref{sec: equilibrium among builders} considers the general setting with $M$ players,
and the two-player game is studied in Section~\ref{sec: two builders}.
Numerical results for the multi-player setting are given in Section \ref{sec: simulation for three builders}.

\subsection{Nash Equilibrium among \texorpdfstring{$M$}{M} Builders} \label{sec: equilibrium among builders}
We make the following assumptions about the game between builders.
\begin{assumption} \label{fi_vi_assumption}
    The parameters satisfy \texorpdfstring{$\bar{f}_i > 0$, $\bar{v}_i > 0$, and $\bar{f}_i \geq \bar{v}_i$}{parameters are positive and satisfy the required inequalities} for all \texorpdfstring{$i \in \left[ M\right]$}. 
\end{assumption}

\begin{remark}
{\em This assumption implies that, in expectation, the MEV each builder can independently extract is at least as large as the additional MEV that may be obtained from the order being auctioned in the OFA. In other words, the orders auctioned in the OFA should not constitute the primary source of MEV in expectation.}
\end{remark}

\begin{lemma}
    The payoff function $\pi_i \left( h_i | h_{-i}\right)$ is strictly concave with respect to builder $i$'s own strategy $h_i$. \label{concavity}
\end{lemma}
\begin{proof}
Let $H = \sum_{j} h_j$ and $H_{-i} =  \sum_{j \neq i} h_j $. It suffices to note that
    \begin{align*}
    \frac{\partial^2 \pi_i}{\partial h_i^2} = -\frac{2 H_{-i} \left( \bar{f_i} H + H_{-i} \left(H_{-i} - \bar{v_i}\right) + h_i \left(H_{-i} + 2 \bar{v_i}) \right) \right)}{H^4} < 0,
\end{align*}
holds for all $h_i \in B_i$ when $\bar{f}_i \geq \bar{v}_i$ (Assumption \ref{fi_vi_assumption}).
\end{proof}

\quad Recall from Section~\ref{sec: modeling--builders} that the strategy space for player \( i \) is given by \( B_i = [r_i^{\min}, r_i^{\max}] \), where $r_i^{\min} > 0$. Under this specification, the existence of a pure-strategy Nash equilibrium follows immediately.

\begin{theorem}
There exists a pure-strategy Nash equilibrium in the game \( \langle \mathcal{I}, (B_i)_{i \in \mathcal{I}}, (\pi_i)_{i \in \mathcal{I}} \rangle \), where \( \mathcal{I} = \left[ M \right] \). \label{thm: existence of NE_N player}
\end{theorem}

\begin{proof}
    By Rosen’s existence theorem~\citep{8344bf6b-b18f-3bf0-88f4-c6a0b490d844}, the existence of a pure-strategy Nash equilibrium follows from the fact that the strategy space is convex and compact, and the utility function \( \pi_i \) is continuous in all \( h_j \) for \( j \in [M] \) and concave in \( h_i \).
\end{proof}

\quad Although Section~\ref{sec: modeling--builders} defines the strategy space for player \( i \) as \( B_i = [r_i^{\min}, r_i^{\max}] \), we may, without loss of generality, assume that \( r_i^{\max} \geq \bar{f}_i + \bar{v}_i \), since values outside the desired range can be assigned zero density. Similarly, we may set $r_i^{\min} = \epsilon$ for all $i \in \left[M \right]$, where $\epsilon > 0$ is an arbitrarily small constant. Moreover, we slightly abuse notation by using \( h_i^* \) to denote both the optimal bid over the original strategy space \( B_i \) and over the domain \( (0, \infty) \). The optimal bid \( h_i^* \), which maximizes \( \pi\left(h_i \mid h_{-i}\right) \) as defined in Eq.~\eqref{eq: utility func} over the domain \( (0, \infty) \), is in fact bounded above by \( \bar{f}_i + \bar{v}_i \), as established in the following lemma.

\begin{lemma}
Let
\begin{equation*}
    h_{i}^* = \arg \max_{h_i > 0} \bar{f}_i \frac{h_i}{\sum_{j=1}^M h_j} + \bar{v}_i \left( \frac{h_i}{\sum_{j=1}^M h_j}\right)^2 - \frac{h_i^2}{\sum_{j=1}^M h_j}.
\end{equation*}
Then it follows that $h_i^* < \bar{f}_i + \bar{v}_i$.
\label{lem: hi<=M}
\end{lemma}

\begin{proof}
    Let $H = \sum_{j=1}^M h_j$ and $H_{-i} = \sum_{j \neq i} h_j$. Consider $\pi\left(h_i \vert h_{-i}\right)$. When $h_i \geq \bar{f}_i + \bar{v}_i$, we have
    \begin{equation*}
        \pi\left(h_i \vert h_{-i}\right) = \frac{1}{H^2} \left[ H_{-i} \left( \bar{f}_i h_i - h_i^2 \right) - h_i^2 \left(  h_i- \bar{f}_i - \bar{v}_i \right) \right] \leq 0.
    \end{equation*}
When $h_i = \bar{f}_i$, we obtain
\begin{equation*}
        \pi\left(h_i = \bar{f}_i \vert h_{-i}\right) = \frac{1}{H^2}   \bar{f}_i^2 \bar{v}_i  > 0.
\end{equation*}
By Lemma \ref{concavity}, since $\pi\left(h_i \vert h_{-i}\right)$ is strictly concave and continuous in $h_i$, it follows that $h_i^* < \bar{f}_i + \bar{v}_i$, as any $h_i^* \geq \bar{f}_i + \bar{v}_i$ cannot be the best response for player $i$.
\end{proof}

\begin{remark}
{\em Based on the above discussion and the preceding lemma, we may, without loss of generality, redefine the strategy space of \( h_i \) as \( B_i = \left[\epsilon, \bar{f}_i + \bar{v}_i\right] \), since the optimal bid \( h_i^* \) over the entire domain \( (0, \infty) \) is bounded above by \( \bar{f}_i + \bar{v}_i \).
}
\end{remark}

\quad The Nash equilibrium of the game can be characterized as follows.

\begin{lemma}
    The first-order conditions for the Nash equilibrium of the game \( \langle \mathcal{I}, (B_i)_{i \in \mathcal{I}}, (\pi_i)_{i \in \mathcal{I}} \rangle \), where \( \mathcal{I} = \left[ M \right] \), are given by:
    \begin{equation*} \label{FOC for NE}
        \frac{\partial \pi_i}{\partial h_i} = 0 \quad \text{for } i \in \left[ M \right],
    \end{equation*}
where
\begin{equation}
    \frac{\partial \pi_i}{\partial h_i} = \frac{- \bar{f}_i h_i H + \bar{f}_i H^2 + 2 \bar{v}_i h_i H - 2 h_i^2 \bar{v}_i - 2 h_i H^2 + h_i^2 H}{H^3} \quad \text{for } i \in \left[ M \right], \label{first derivative of pi}
\end{equation}
with $H = \sum_{j} h_j$. 

\quad If the system of equations defined by the first-order conditions admits at least one positive solution, let $\left(h_1^\prime, h_2^\prime, \dots, h_m^\prime \right)$ denote such a solution. Define $h_i^* = \max \left\{ h_i^\prime, \epsilon \right\}$ for each $i \in \left[M \right]$. Then $\left(h_1^*, h_2^*, \dots, h_m^* \right)$ is the pure-strategy Nash equilibrium of the game. If no positive solution exists, then the pure-strategy Nash equilibrium must lie on the boundary of the domain $\left[\epsilon, C_1 \right] \times \left[\epsilon, C_2 \right] \times \ldots \times \left[\epsilon, C_M \right]$ where $C_i = \bar{f}_i+\bar{v}_i$.
\end{lemma}

\quad In the remaining of this section, 
we focus on the case where $M = 2$.

\subsection{\texorpdfstring{$M = 2$}{M=2} Builders} \label{sec: two builders}

 Given that \(h_1, h_2 \neq 0\), let \(h_2 = \lambda h_1\), where \(\lambda > 0\). Substituting this into Eq.~(\ref{first derivative of pi}), the equations simplify to:
    \begin{align}
    \begin{cases}
    h_1^3 \left(1+2\lambda\right)\left(1+\lambda\right) = \lambda h_1^2 \left(\bar{f}_1 \lambda + \bar{f}_1 + 2 \bar{v}_1\right), \\
    h_1^3 \lambda \left(2 + \lambda\right)\left(1+\lambda\right) = h_1^2 \left(\bar{f}_2 + (\bar{f}_2 + 2 \bar{v}_2\right)\lambda).
    \end{cases} \label{equiv FOC}
    \end{align}
    Solving Eq.~(\ref{equiv FOC}) is equivalent to solving the quartic equation:
    \begin{equation}
    P(\lambda) = \bar{f}_1 \lambda^4 + \left(3 \bar{f}_1 + 2\bar{v}_1\right)\lambda^3 + \left(2\bar{f}_1-2\bar{f}_2+4\bar{v}_1-4\bar{v}_2\right) \lambda^2 - \left(3\bar{f}_2 + 2\bar{v}_2\right)\lambda - \bar{f}_2 = 0. \label{FOC: quartic equation}
    \end{equation}
    Once $\lambda$ is obtained, we set $h_2 = \lambda h_1$ and solve for $h_1$ using either of the reformulated equations:
    \begin{equation}
    h_1 = \dfrac{\lambda \left(\bar{f}_1 \lambda + \bar{f}_1 + 2 \bar{v}_1\right)}{\left(1+2\lambda\right)\left(1+\lambda\right)} = \dfrac{\bar{f}_2 + \left(\bar{f}_2 + 2 \bar{v}_2\right) \lambda}{\lambda \left(2 + \lambda\right) \left(1 + \lambda\right)}.  \label{solve for h1 from lamda}
    \end{equation}

\begin{theorem}
The game \( \langle \mathcal{I}, (B_i)_{i \in \mathcal{I}}, (\pi_i)_{i \in \mathcal{I}} \rangle \), where \( \mathcal{I} = \{1, 2\} \), has a unique  Nash equilibrium \( h^{\text{NE}} = \left(h_1^{\text{NE}}, h_2^{\text{NE}} \right) \). \label{uniquess of NE}
\end{theorem}


\begin{proof}    
    We analyze the solutions to Eq.~(\ref{FOC: quartic equation}). By Vieta's formulas, the products of the roots satisfies $\lambda_1 \lambda_2  \lambda_3 \lambda_4 = -\frac{\bar{f_2}}{\bar{f_1}} < 0$, indicating the roots of Eq.~(\ref{FOC: quartic equation}) must satisfy one of the following: (1) one positive real root and three negative real roots, (2) three positive real roots and one negative real root, or (3) one positive real root, one negative real root, and two complex real roots. We will eliminate the cases (2) and (3), thereby proving the existence of a unique positive real root.
    
    \quad Assume, for contradiction, that the roots consist of three positive positive real roots, denoted $\lambda_1, \lambda_2, \lambda_3$, and one negative real root, denoted $\lambda_4$. We further observe from Vieta's formulas that
    \begin{align}
        \begin{cases}
            \lambda_1+\lambda_2+ \lambda_3+\lambda_4 = -3 -\frac{2\bar{v_1}}{\bar{f_1}} < -3, \\
            \frac{1}{\lambda_1}+\frac{1}{\lambda_2}+\frac{1}{\lambda_3}+\frac{1}{\lambda_4} = -3 - \frac{2\bar{v_2}}{\bar{f_2}} < -3.
        \end{cases} \notag
    \end{align}
    Since $\lambda_1+\lambda_2+ \lambda_3+\lambda_4 < 0$, it follows that $\lambda_3 + \lambda_4 <0$. However, this implies $\frac{1}{\lambda_3} + \frac{1}{\lambda_4} > 0$, contradicting $\frac{1}{\lambda_1}+\frac{1}{\lambda_2}+\frac{1}{\lambda_3}+\frac{1}{\lambda_4} < 0$. This contradiction eliminates the case (2).
    
    \quad To eliminate the case (3), we observe that $P(0) =  - \bar{f}_2 < 0$, $P(-\frac{1}{2}) = \frac{3}{16} \bar{f}_1 + \frac{3}{4} \bar{v}_1 > 0$, and $P(-2) = -3\bar{f}_2-12\bar{v}_2 < 0$. Since $P(\lambda)$ is continuous in $\lambda$, there exists one root in the interval $\left(-\frac{1}{2}, 0\right)$ and another root in the interval $\left(-2, -\frac{1}{2}\right)$ by the intermediate value theorem. Thus, the polynomial has at least two negative real roots, leaving no possibility for two complex roots.

    \quad Having ruled out the cases (2) and (3), we conclude that the only remaining possibility is the case (1), where there exists exactly one positive real root, and it is unique.

    \quad We now prove that the unique positive solution to Eq.~(\ref{FOC: quartic equation}), denoted by $\lambda^*$, corresponds to the unique Nash equilibrium of the game. Let $h_1^* = \dfrac{\lambda^* \left(\bar{f}_1 \lambda^* + \bar{f}_1 + 2 \bar{v}_1\right)}{\left(1+2\lambda^*\right)\left(1+\lambda^*\right)}$, $h_2^* = \lambda^* h_1^*$, and $h_i^{\text{NE}} = \max \left\{ h_i^*, \epsilon \right\}$ for $i=1,2$. By Lemma~\ref{concavity}, the strategy profile $h^{\text{NE}} =\left(h_1^{\text{NE}}, h_2^{\text{NE}}\right)$ is a Nash equilibrium. Since the optimal solution \(h^* = \left( h_1^*, h_2^*\right)\) over the domain \(\left(0, \infty\right)\) is unique, the resulting equilibrium \( \left(h_1^{\text{NE}}, h_2^{\text{NE}}\right) \) is also unique.
    
\end{proof}

\quad After establishing the existence and uniqueness of the Nash equilibrium, we derive its closed-form solution. The proof relies on a computer-assisted technique.
\begin{theorem}
    The unique Nash equilibrium is given by:
    \begin{equation}
        \begin{aligned}
        \lambda^* &= -\frac{3 \bar{f_1} + 2 \bar{v_1}}{4 \bar{f_1}} + S + \frac{1}{2} \sqrt{-4S^2-2p - \frac{q}{S}} \label{lamda},\\
        h_1^* &= \dfrac{\lambda^* (\bar{f_1} \lambda^* + \bar{f_1} + 2 \bar{v_1})}{(1+2\lambda^*)(1+\lambda^*)} = \dfrac{\bar{f_2} + (\bar{f_2} + 2 \bar{v_2}) \lambda^*}{\lambda^* (2 + \lambda^*) (1 + \lambda^*)}, \\
        h_2^* &= \lambda^* h_1^*,\\
        h_i^{\text{NE}} &= \max \left\{ h_i^*, \epsilon \right\}, \text{ for } i = 1,2,
        \end{aligned}
    \end{equation}
    where 
    \begin{equation}
        \begin{aligned}
        p &= -\frac{11 \bar{f_1}^2 + 12 \bar{v_1}^2 + 4\bar{f_1}(4\bar{f_2} + \bar{v_1} + 8 \bar{v_2})}{8 \bar{f_1}^2},\\
        q &= \frac{3 \bar{f_1}^3 + 8 \bar{v_1}^3 + 4 \bar{f_1} \bar{v_1}(4 \bar{f_2} + \bar{v_1} + 8 \bar{v_2}) + \bar{f_1}^2 (-10 \bar{v_1} + 32 \bar{v_2})}{8 \bar{f_1}^3},\\
        \Delta_0 &= (2 \bar{f_1}-2 \bar{f_2}+4 \bar{v_1}-4 \bar{v_ 2})^2+3 (3 \bar{f_1}+2 \bar{v_1})(3 \bar{f_2}+2 \bar{v_2})-12 \bar{f_1} \bar{f_2}, \\
        \Delta_1 &= -27 \bar{f_2} (3 \bar{f_1} + 2 \bar{v_1})^2 + 72 \bar{f_1} \bar{f_2} (2 \bar{f_1} - 2 \bar{f_2} + 4 \bar{v_1} - 4 \bar{v_2} ) + 2 (2 \bar{f_1} - 2 \bar{f_2} + 4 \bar{v_1} - 4 \bar{v_2})^3 \\& \quad + 9(3 \bar{f_1} + 2 \bar{v_1})(2 \bar{f_1} - 2 \bar{f_2} + 4 \bar{v_1} - 4 \bar{v_2})(3 \bar{f_2} + 2 \bar{v_2}) + 27 \bar{f_1} (3 \bar{f_2} + 2 \bar{v_2})^2, \\
        \varphi &= \arccos{\left(\frac{\Delta_1}{2 \sqrt{\Delta_0^3}}\right)},\\
        S &= \frac{1}{2} \sqrt{-\frac{2}{3}p + \frac{2}{3 \bar{f_1}} \sqrt{\Delta_0} \cos{\frac{\varphi}{3}}}.
        \end{aligned} \label{p,q,S}
    \end{equation}
    
\end{theorem}

\begin{proof}
Since all the roots of Eq.~(\ref{FOC: quartic equation}) are real and exactly one of them is positive, the four real roots can be expressed as follows:
\begin{align}
    \lambda_{1,2} &= -\frac{3 \bar{f_1} + 2 \bar{v_1}}{4 \bar{f_1}} - S \pm \frac{1}{2} \sqrt{-4S^2-2p + \frac{q}{S}}, \notag\\ 
    \lambda_{3,4} &= -\frac{3 \bar{f_1} + 2 \bar{v_1}}{4 \bar{f_1}} + S \pm \frac{1}{2} \sqrt{-4S^2-2p - \frac{q}{S}}. \notag
\end{align}
where $p,q,S$ are defined in Eqs.~(\ref{p,q,S}). We want to show that $- S + \frac{1}{2} \sqrt{-4S^2-2p + \frac{q}{S}} < S + \frac{1}{2} \sqrt{-4S^2-2p - \frac{q}{S}}$, which identifies the positive (and largest) real root of the quartic equation as Eq.~(\ref{lamda}).
To prove this, observe that:
\begin{equation}
    \begin{aligned}
   & - S + \frac{1}{2} \sqrt{-4S^2-2p + \frac{q}{S}} < S + \frac{1}{2} \sqrt{-4S^2-2p - \frac{q}{S}}, \\
    \Longleftrightarrow & \sqrt{-4S^2-2p + \frac{q}{S}} < \ 4S +\sqrt{-4S^2-2p + \frac{q}{S}}, \\
    \Longleftrightarrow & 8S^3 - q + 4S^2\sqrt{-4S^2-2p - \frac{q}{S}} > 0.\label{last inequality for largest real root}
    \end{aligned}
\end{equation}

Here we rely on numerical computations. We find that the optimal value of the following minimization problem is approximately $3 \ (>0)$:
\begin{equation}
\begin{aligned}
    \min_{\bar{f_1}, \bar{f_2}, \bar{v_1}, \bar{v_2} \in \mathbb{R^+}} \quad & 8S^3 - q, \\
    \text{such that} \quad & \bar{f_1} \geq \bar{v_1}, \\
                        & \bar{f_2} \geq \bar{v_2}.
\end{aligned}
\label{optimization problem (solutions)}
\end{equation}
(The implementation details and code are provided in Appendix~\ref{appendix: codes}.)
This ensures that the inequality~(\ref{last inequality for largest real root}) holds,
and the theorem is proved.
\end{proof}


\quad We then explore how $h_1^*$ and $h_2^*$ vary when we change the parameters $\bar{f}_i$ and $\bar{v}_i$ for $i = 1,2$. Specifically, we examine whether, when one player is $k$ times more capable than the other in extracting MEV, the ratio of their optimal payments also equals $k$.

\begin{proposition}
    Let $k_1, k_2$ be given such that $\frac{\bar{v_1}}{\bar{f_1}} = \frac{\bar{v_2}}{\bar{f_2}} = k_1$ and $\frac{\bar{f_1}}{\bar{f_2}} = k_2$, where $0 < k_1 \leq 1$ as given by Assumption~\ref{fi_vi_assumption}. Then the following statements hold: \begin{itemize}
        \item If $k_2 < 1$, then $\frac{h_1^*}{h_2^*}>k_2$.
        \item If $k_2 = 1$, then $\frac{h_1^*}{h_2^*} = k_2 = 1$, and $h_1^*$ = $h_2^* = \frac{\bar{f_1} + \bar{v_1}}{3} = \frac{\bar{f_2} + \bar{v_2}}{3}$.
        \item If $k_2>1$, then $\frac{h_1^*}{h_2^*} < k_2$.
    \end{itemize}
    \label{prop: 2 player NE ratios}
\end{proposition}

\begin{proof}
    Since $\frac{h_1}{h_2} = \frac{1}{\lambda}$, proving the theorem is equivalent to proving that the largest real root of Eq.~(\ref{FOC: quartic equation}), denoted by $\lambda_4$, satisfies
    \begin{equation*}
        \lambda_4 <  \frac{1}{k_2} \text{ for } k_2 < 1, \ \lambda_4 =  \frac{1}{k_2} = 1 \text{ for } k_2 = 1, \ \lambda_4 >  \frac{1}{k_2} \text{ for } k_2 > 1. 
    \end{equation*} 
    Given that Eq.~(\ref{FOC: quartic equation}) satisfies $P(0) = -\bar{f}_2 < 0$, it follows that proving the desired inequalities for $\lambda_4$ is equivalent to verifying that $P(\frac{1}{k_2})$ is positive for $k_2 < 1$, zero for $k_2 = 1$, and negative for $k_2 > 1$. Since $P(\lambda)$ is continuous in $\lambda$ and $P(0) < 0$, it suffices to establish the sign of $P(\frac{1}{k_2})$.
    
    \quad Rewriting $P(\lambda)$ in terms of $k_1, k_2$, we obtain:
    \begin{equation*}
        P(\lambda;k_1, k_2) = \bar{f} _2\left[k_2 \lambda^4 + \left( 3 k_2 + 2 k_1 k_2\right) \lambda^3 + \left( 2 k_2 -2 + 4 k_1k_2-4 k_1\right) \lambda^2 - \left( 3+2k_1 \right)\lambda - 1 \right] = 0 \label{FOC:quartic equation with k1, k2}.
    \end{equation*}
    Define $G(k_1,k_2)$ as the evaluation of $P(\lambda;k_1,k_2)$ at $\lambda = \frac{1}{k_2}$.
    \begin{equation*}
        G(k_1, k_2) = P(\lambda =\frac{1}{k_2};k_1,k_2) = \frac{\bar{f}_2}{k_2^3} \left[1 + \left( 2k_2+1+4k_1k_2-2k_1\right) k_2 - \left( 3+2k_1 \right) k_2^2 - k_2^3 \right].
    \end{equation*}
    Differentiating with respect to $k_1$, we obtain
    \begin{equation*}
        \frac{\partial G(k_1, k_2)}{\partial k_1} = \frac{2 \bar{f}_2  (k_2 - 1)}{k_2^2}.
    \end{equation*}
    Since $\frac{\partial G(k_1, k_2)}{\partial k_1} < 0$ for all $k_1 \in \left( 0,1 \right]$ when $k_2 \in \left( 0,1 \right)$, it follows that $G(k_1,k_2) > G(1, k_2)$ for all $k_1 \in \left(0, 1\right)$. Evaluating $G(1,k_2)$, we find
    \begin{equation*}
        G(1,k_2) = \frac{\bar{f}_2}{k_2^3} \left( 1 + k_2^2 - k_2 - k_2^3 \right) = \frac{\bar{f}_2}{k_2^3}  \left( 1-k_2 \right) \left( 1+k_2^2 \right),
    \end{equation*}
    which is positive for $k_2 \in \left( 0,1\right)$. Therefore, $G(k_1, k_2) > 0$ for all $k_1 \in \left( 0,1 \right]$ and $k_2 \in \left( 0,1 \right)$.
    
    \quad Similarly, for $k_2 > 1$, we have $\frac{\partial G(k_1, k_2)}{\partial k_1} > 0$ for all $k_1 \in \left( 0,1 \right]$, implying that $G(k_1,k_2) < G(1, k_2)$ for all $k_1 \in \left(0, 1\right)$. Since $G(1,k_2) = \frac{\bar{f}_2}{k_2^3}  \left( 1-k_2 \right) \left( 1+k_2^2 \right) < 0$ for $k_2 > 1$, it follows that $G(k_1, k_2) < 0$ in this case. 
    
    \quad For $k_2 = 1$, direct substitution yields $G(k_1,k_2 = 1)= 0$, which implies that $P(\lambda = 1;k_1,k_2) = 0$, and hence $\lambda = 1$ is a root of $P(\lambda)$. Consequently, it follows that $h_1 = h_2$. Additionally, when $k_2 = 1$, we have $\bar{f}_1 = \bar{f}_2$, $\bar{v}_1 = \bar{v}_2$. Substituting these identities into Eq.~(\ref{FOC for NE}), we obtain:
    \begin{equation*}
        h_1 = h_2 = \frac{\bar{f}_1 + \bar{v}_1}{3} = \frac{\bar{f}_2 + \bar{v}_2}{3}.
    \end{equation*}
    This completes the proof.
\end{proof}

\quad Proposition~\ref{prop: 2 player NE ratios} shows that when a player's ability to extract MEV, both with and without the order being auctioned in the order flow auction, is $k$ times that of the other player (in expectation), their equilibrium payments do not scale proportionally with $k$. In particular, if player $2$ can obtain more MEV, and this amount is $k$ times that of player $1$ in expectation, then player $2$’s optimal payment is less than $k$ times the payment of player $1$.

\quad For example, consider the case where $\bar{f}_1 = 100, \bar{f}_2 = 200, \bar{v}_1= 40$, and $\bar{v}_2 =80$. At equilibrium, the corresponding bids are $h_1 = 49.94$ and $h_2 = 82.17$. After a single round of the game, player $1$'s expected revenue is $24.64$, while player $2$'s expected revenue is $104.23$. Although player $2$’s MEV extraction capability is only twice that of player $1$, they ultimately earn more than four times player $1$’s revenue, in expectation. This illustrates that the gap between their MEV extraction abilities is amplified by the interaction between the order flow auction and the block-building auction.

\quad For $M=3$ and $M=4$, we conduct numerical experiments to study the equilibrium solutions, as detailed in Section~\ref{sec: simulation for three builders}.

\section{Evolution of Validators’ Stake Shares} \label{sec: validators' shares}

\quad In this section, we analyze the dynamics of validators' stake shares over time. 
Each time a validator is selected to propose a block, their reward comes either from the bid submitted by the winning builder or from the MEV they extract by constructing the block themselves.
\begin{theorem} \label{thm: martingale--validator}
    Validator $j$'s stake share, $\left(\omega_{j,t}\right)_{t \geq 0}$, is a martingale, and the limit $\omega_{j, \infty}\coloneqq \lim_{t \to \infty} \omega_{j,t}$ exists almost surely, with $\mathbf{E} \left( \omega_{j,\infty} \right) = \omega_{j,0}$.
\end{theorem}
\begin{proof}
For each $j$,
    \begin{equation}
    \label{eq:wjtx}
    \begin{aligned}
         \omega_{j,t+1} &= \frac{s_{j,t+1}}{S_{t+1}},  \\
        &=  \frac{s_{j, t} + R_{t+1} \mathbf{1}\left(X_{j,t+1}\right) - \alpha \frac{s_{j,t}}{S_{t}^{1+\gamma}}}{S_{t+1}}, \\
        &= \omega_{j,t} \frac{S_t - \alpha / {S_t^{\gamma}}}{S_{t+1}} + \frac{R_{t+1}}{S_{t+1}} \mathbf{1}\left(X_{j,t+1}\right) .
    \end{aligned}
\end{equation}
Noting that $\mathbf{E} \left[\mathbf{1} \left(X_{j,t+1}\right)|\mathcal{F}_t\right] = \omega_{j,t}$, we have
\begin{equation}
    \begin{aligned}
        \mathbf{E} \left[\omega_{j,t+1} | \mathcal{F}_t\right] &= \omega_{j,t} \mathbf{E} \left[\frac{S_t-\alpha / S_t^{\gamma}+R_{t+1}}{S_{t+1}} \middle| \mathcal{F}_t \right], \\
        &= \omega_{j,t} .\label{martingale w_jt}
    \end{aligned}
\end{equation}
where Eq.~(\ref{martingale w_jt}) holds because $S_{t+1} = S_t - \alpha / S_t^{\gamma} + R_{t+1}$. This shows that $\left(\omega_{j,t}\right)_{t\geq0}$ is a martingale. Since it's a non-negative martingale, the martingale convergence theorem ensures that the limit $\omega_{j,\infty}:=\lim_{t \to \infty} \omega_{j,t}$ exists almost surely. Furthermore, as $\omega_{j,t}$ is bounded, the bounded convergence theorem implies that $\mathbf{E} \left( \omega_{j,\infty}\right) = \lim_{t \to \infty} \mathbf{E} \left( \omega_{j,t} \right) = \omega_{j,0} $. 
\end{proof}

\quad This proposition suggests that centralization is unlikely to occur in the validator space, as their stake shares follow a martingale process and depend  only on their initial shares. We next analyze the evolution of the total stake controlled by the validators in the system. Recall from Section~\ref{sec: modeling--validators} that the coefficient for the cost of the system, $\alpha$, is linked to the builder's bid through Assumption~\ref{assump: bound on alpha}.

\begin{lemma}
    Under Assumption~\ref{assump: bound on alpha}, the sequence $\{S_t\}_{t \ge 1}$ is strictly increasing. \label{lemma: S_t increasing}
\end{lemma}

\begin{proof}
    We prove the lemma by induction. At $t=1$, we have
    \begin{equation*}
        S_1 = S_0 + R_1 - \frac{\alpha}{S_0^\gamma}.
    \end{equation*}
 By Assumption~\ref{assump: bound on alpha}, it follows that $\frac{\alpha}{S_0^\gamma} < \min\left\{ k_1, k_2, \ldots, k_M \right\}$ where $k_i = \max\left\{ \beta_v, \mu r_i^{\min} \right\}$ for each $i \in \left[ M\right]$. Since $R_1 \geq R_{\min}$, we conclude that $\frac{\alpha}{S_0^\gamma} < \min\left\{ k_1, k_2, \ldots, k_M \right\} \leq R_1$, which implies $S_1 > S_0$. Suppose that $S_0 < S_1 < \ldots < S_k$ holds for some $k > 1$. At $t = k+1$, 
\begin{equation*}
    S_{k+1} = S_k + R_k -\frac{\alpha}{S_k^\gamma}.
\end{equation*}
Similarly, $\frac{\alpha}{S_k^\gamma} < \left(\frac{S_0}{S_k}\right)^\gamma \cdot \min\left\{ k_1, k_2, \ldots, k_M \right\} < R_k$ since $S_0 < S_k$. Thus, $S_{k+1} > S_k$. This completes the induction and proves the lemma. 
\end{proof}

\begin{proposition} \label{prop: long-time behavior of St} (Long-time behavior of $S_t$)
The following results hold:
\begin{enumerate}
    \item The process $\left( S_t, t \geq 0 \right)$ is an $\mathcal{F}_t$-sub-martingale, and its compensator is
    \begin{equation*}
        A_t = Rt - \alpha\sum_{k=0}^{t-1} \frac{1}{S_k^\gamma}.
    \end{equation*}
    \item There is the convergence in probability:
        \begin{enumerate}
            \item If $\gamma > 0$, then
            \begin{equation*}
                \frac{S_t}{t} \to R \hspace{1em} \text{ as } t \to \infty.
            \end{equation*}
            \item If $\gamma = 0$, then
            \begin{equation*}
                \frac{S_t}{t} \to  R -\alpha \hspace{1em} \text{ as } t \to \infty.
            \end{equation*}
        \end{enumerate}
\end{enumerate}
\end{proposition}

\begin{proof}
(1) It suffices to note that $\mathbf{E}\left( S_{t+1} | \mathcal{F}_t \right) = S_t + R -\frac{\alpha}{S_t^\gamma}$.

(2) Apply the method of moments by computing $\mathbf{E} \left( S_{t}^{k} \right)$ for all $k$. For $k=1$, we have by definition:
    \begin{align*}
        \mathbf{E}\left( S_{t+1} - S_{t} | S_t = s\right) &= R - \frac{\alpha}{s^{\gamma}}, \\
        &= \begin{cases}
            R & \text{if } \gamma >0\\
            R - \alpha & \text{if } \gamma = 0
        \end{cases}
        \quad \text{as } s \to \infty .
    \end{align*}
It is clear that with probability one $S_t \to \infty$ as $t \to \infty$ when $\gamma \geq 0$. As a result, 
\begin{equation*}
\begin{cases}
    \mathbf{E}\left( S_{t+1} - S_{t} \right) \to R \quad \text{as } t \to \infty & \text{ if } \gamma > 0, \\ \mathbf{E}\left( S_{t+1} - S_{t} \right) \to R - \alpha \quad \text{as } t \to \infty & \text{ if } \gamma = 0.
\end{cases}
\end{equation*}
which yields
\begin{equation*}
\begin{cases}
     \mathbf{E} \left( S_t \right) \sim Rt \quad \text{as } t\to \infty & \text{ if } \gamma > 0,\\
    \mathbf{E} \left( S_t \right) \sim \left(R-\alpha\right) t  \quad \text{as } t\to \infty  & \text{ if } \gamma = 0.
\end{cases}
\end{equation*}
Next for $k = 2$, we have:
    \begin{align*}
        \mathbf{E}\left( S_{t+1}^2 - S_{t}^2 | S_t = s\right) &=
       \mathbf{E} \left(R_{t+1}^2\right) + \frac{\alpha^2}{s^{2\gamma}} + 2sR -2\frac{\alpha s}{s^{\gamma}}-2 \frac{\alpha}{s^{\gamma}} R , \\
       &=\begin{cases}
           \mathbf{E} \left(R_{t+1}^2\right) + 2Rs + \mathcal{O}(s^{1-\gamma}) & \text{if } \gamma>0 \\
           \mathbf{E} \left(R_{t+1}^2\right) + \alpha^2 + 2\left(R-\alpha\right)s -2\alpha R & \text{if } \gamma=0
       \end{cases}
       \quad \text{ as } s \to \infty.
    \end{align*}
Thus, 
\begin{align*}
\begin{cases}
    \mathbf{E}\left( S_{t+1}^2 - S_{t}^2\right) =  \left( 2R + o(1) \right) \mathbf{E}\left(S_{t}\right) \sim 2R^2t & \text{if } \gamma > 0, \\
    \mathbf{E}\left( S_{t+1}^2 - S_{t}^2\right) = \left( 2\left(R-\alpha\right) + o(1) \right) \mathbf{E}\left(S_{t}\right) \sim 2\left(R-\alpha\right)^2t & \text{if } \gamma=0.
\end{cases}
\end{align*}
Then we get:
\begin{equation*}
    \begin{cases}
        \mathbf{E}\left( S_{t}^2 \right) \sim R^2 t^2 & \text{ if } \gamma > 0\\
        \mathbf{E}\left( S_{t}^2 \right) \sim \left(R-\alpha\right)^2 t^2 & \text{ if } \gamma = 0
    \end{cases}
    \quad \text{ as } t \to \infty.
\end{equation*}
We proceed by induction. Assume that 
\begin{equation*}
\begin{cases}
    \mathbf{E}\left( S_{t}^k \right) \sim R^k t^k & \text{ if } \gamma > 0 \\
    \mathbf{E}\left( S_{t}^k \right) \sim \left(R-\alpha\right)^k t^k & \text{ if } \gamma = 0
\end{cases}
\quad \text{ as } t \to \infty.
\end{equation*}
We obtain:
\begin{equation*}
\begin{cases}
 \mathbf{E}\left( S_{t+1}^{k+1} - S_{t}^{k+1}\right) =  \left( \left(k+1\right)R + o(1) \right) \mathbf{E}\left(S_{t}^{k}\right) \sim \left(k+1\right)R^{k+1} t^k & \text{ if } \gamma > 0, \\ \mathbf{E}\left( S_{t+1}^{k+1} - S_{t}^{k+1}\right) = \left( \left(k+1\right) \left(R-\alpha\right) + o(1) \right) \mathbf{E}\left(S_{t}^k\right) \sim \left(k+1\right) \left(R-\alpha\right)^{k+1} t^k & \text{ if } \gamma=0. 
 \end{cases}
 \end{equation*}
 Thus, we have:
\begin{equation*}
\begin{cases}
    \mathbf{E} \left( S_t^{k+1} \right) \sim R^{k+1} t^{k+1} \quad \text{as } t\to \infty & \text{ if } \gamma > 0, \\ \mathbf{E} \left( S_t^{k+1} \right) \sim \left(R-\alpha\right)^{k+1} t^{k+1}  \quad \text{as } t\to \infty & \text{ if } \gamma = 0.
    \end{cases}
\end{equation*}
By the method of moments, $S_t/t$ converges in distribution, and thus in probability to $R$ if $\gamma > 0$, and to $R - \alpha$ if $\gamma = 0$.  
\end{proof}

\quad In general, it is difficult to make explicit the distribution of $\omega_{j, \infty}$.
Nevertheless, in the case where $\alpha = 0$, the following proposition characterizes the distribution of $\omega_{j, \infty}.$

\quad Recall from Section~\ref{sec: modeling--validators} that $\left\{R_t\right\}_{t \geq 1}$ is a sequence of i.i.d random variables, where 
\[
R_t = \max \left\{ \mu r_i, \beta_v \right\},
\]
with probability 
\[
p_i \coloneqq \frac{h_i}{\sum_{k=1}^M h_k} \quad \text{ for } i \in \left[ M \right].
\]


Note that $r_i$ is a random variable. Let $u$ denote the probability density function of $R_t$.
Let $\mathcal{K} = \left[R_{\min}, R_{\max} \right]$ denote the support of $R_t$ for all $t \geq 1$. Recall from Section~\ref{sec: modeling--validators} that $R_{\min} \coloneqq \min\left\{ k_1, k_2, \ldots, k_M \right\}$ and $R_{\max} \coloneqq \max\left\{ g_1, g_2, \ldots, g_M \right\}$, where $k_i = \max\left\{ \beta_v, \mu r_i^{\min} \right\}$ and $g_i = \max\left\{ \beta_v, \mu r_i^{\max} \right\}$ for each $i \in \left[ M\right]$. Let $\mathscr{P} \left(\left[ 0,1\right]\right)$ denote the space of distribution functions with support in $\left[0,1 \right]$, and $\mathcal{S} \coloneqq \left[0, \infty \right) \times \left[0, \infty \right) \backslash \left\{ \left( 0,0\right) \right\}$. 
For any $x \in \mathbb{R}$, let $\delta_x$ be the distribution of the point mass at $x$. Define the function
\begin{equation*}
    \mathbf{F}_j : \mathcal{S} \to \mathscr{P} \left(\left[ 0,1\right]\right),
\end{equation*}
that maps the initial stake pair $\left(s_{j,0}, S_0 - s_{j,0}\right)$, where $S_0 = \sum_{j=1}^N s_{j,0}$, to the probability distribution $\mathbf{F}_j \left( s_{j,0}, S_0 - s_{j,0} \right)$ of the limiting stake share $\omega_{j, \infty}$. Recall that the initial stake share is given by $\omega_{j,0} = \frac{s_{j,0}}{S_0}$.

\begin{proposition} \label{prop: dist of limit of shares}
    When $\alpha = 0$, for all $\left(s_{j,0}, S_0-s_{j,0} \right) \in \mathcal{S}$, the distribution function of $\omega_{j, \infty}$ satisfies
    \begin{eqnarray} \label{eq: limiting dist characterization}
    \mathbf{F}_j\left(s_{j,0}, S_0 - s_{j,0}\right) &=& \omega_{j,0} \int_{\mathcal{K}} \mathbf{F}_j  \left( s_{j,0}+m, S_0 - s_{j,0} \right) u\left(m\right)dm  \nonumber\\
    &+&\left( 1- \omega_{j,0}\right) \int_{\mathcal{K}} \mathbf{F}_j  \left( s_{j,0}, S_0 - s_{j,0}+m \right) u\left(m\right)dm,
    \end{eqnarray}
    and it is the unique solution to Eq.~(\ref{eq: limiting dist characterization}) among the continuous functions $\mathscr{G}: \mathcal{S} \to \mathscr{P} \left(\left[ 0,1\right]\right)$ satisfying the following three conditions:
    \begin{enumerate}
        \item $\mathscr{G} \left(0,a \right) = \delta_0$ for $a>0$;
        \item $\mathscr{G} \left(a,0\right) = \delta_1$ for $a>0$;
        \item For every $\epsilon > 0$, there exists a $C = C \left( \epsilon \right)$ such that
        \begin{equation*}
            d_w\left( \mathscr{G} \left(s_{j,0},S_0 - s_{j,0} \right), \delta_{\omega_{j,0}} \right) < \epsilon,
        \end{equation*}
        if $S_0 > C$, where 
        \begin{equation*}
            d_w \left(F,G\right) = \int_0^1 |F(x) - G(x)|dx \quad\text{ for all } F,G \in \mathscr{P} \left(\left[ 0,1\right]\right). 
        \end{equation*}
    \end{enumerate}
\end{proposition}

\begin{proof}
    When $\alpha = 0$, the stake share evolution of each validator $j$ with initial stake $s_{j,0}$ can be modeled as a Pólya urn model with random replacements. Specifically, we consider an urn containing two types of balls, black and white, where the initial number of black balls is $s_{j,0}$ and the initial number of white balls is $S_0 - s_{j,0}$. The remainder of the proof follows~\citet{giacomo2007urn}. 
\end{proof}

\quad Refer to Section~\ref{sec: simulation for validators' shares} for the numerical results on the distribution of $\omega_{j, \infty}$. Comparing Figure~\ref{fig: final stake share dist a>0} and Figure~\ref{fig: final stake share dist a=0}, we observe that 
when $\gamma > 0$,
the distribution of $\omega_{j, \infty}$ appears to be very similar for
$\alpha = 0$ and $\alpha \ne 0$.  


\quad When $\gamma = 0$, we can compute explicitly the variance of $\omega_{j, \infty}$, and study its stability. 
The following theorem characterizes its asymptotics.


\begin{theorem} \label{thm: stability of share}
    Let $R_{\min}$ and $R_{\max}$ denote the minimum and maximum values, respectively, of the i.i.d random variable $R_t$. For $s_{j,0} = f\left(S_0\right)$ such that $f\left(S_0\right) \to \infty$ as $S_0 \to \infty$, we have for each $\epsilon > 0$ and each $t \geq 1$ or $t = \infty$:

    \begin{equation}
        \mathbf{P} \left( \left\lvert \frac{\omega_{j,t}}{\omega_{j,0}} - 1\right\lvert > \epsilon \right) \leq \frac{R_{\max}^2}{\left( R_{\min} - \alpha\right)\epsilon^2 f\left( S_0\right)}, \label{ineq: limiting share cvg to initial share}
    \end{equation}
    which converges to $0$, as $S_0 \to \infty$.
\end{theorem}

\quad The proof of Theorem~\ref{thm: stability of share} is given in Appendix~\ref{proof: thm-stability}. 
It implies that for large validators, i.e., those with initial stakes $s_{j,0} = f(S_0)$ such that $f(S_0) \to \infty$ as $S_0 \to \infty$, their shares remain stable over time. Specifically, their limiting share converges in probability to their initial share as the total initial stake $S_0 \to \infty$. 


\section{Numerical Results} \label{sec: numerical results}
\quad This section presents numerical results on the strategic interactions among builders and the evolution of validators' stake shares. Section~\ref{sec: simulation for three builders} reports numerical findings for the multi-player setting ($M \geq 3$) in the builders' game, and Section~\ref{sec: simulation for validators' shares} provides simulation results on the dynamics of validators' stake shares.

\subsection{Multi-Player Game among Builders \texorpdfstring{$\left(M \geq 3\right)$}{M larger than or equal to three}} \label{sec: simulation for three builders}
We conduct numerical experiments to analyze the equilibrium behavior in settings with more than two players. In this analysis, we consider two cases:
\begin{enumerate}
    \item Each player's ability to extract MEV without the order being auctioned in the OFA is $k_{1,i} \coloneqq \frac{\bar{v}_i}{\bar{f}_i}$ times the MEV obtained from the auctioned order, and this ratio $k_{1,i}$ remains constant across all three players, i.e., $k_{1,1} = k_{1,2} = k_{1,3}$.
    \item The ratio $k_{1,i}$ varies among the three players, such that $k_{1,1} < k_{1,2} < k_{1,3}$.
\end{enumerate}
The results for four players are presented in Table~\ref{table: 4 player with common ratio} and Table~\ref{table: 4 player with descending ratio} in Appendix~\ref{appendix: additional numerical results}.
\begin{table}[h] 
\centering
\begin{tabular}{ |c|c|c|c|c|c| } 
 \hline
 \rule{0pt}{10pt} $\bar{f}_i/\bar{v}_i$ & 2 & 3 & 5 & 8 & 10 \\ 
 \hline 
 $h_1$ & 138.55 & 184.96 &278.07 & 417.99 & 511.34\\ 
 $\mathbf{E} \left(\pi_1\right)$ &147.81 &203.40 &315.73 & 484.98& 597.97 \\
 \hline
 $h_2$& 64.31 & 89.42 & 139.38 & 214.14 & 263.93\\ 
 $\mathbf{E} \left(\pi_2\right)$ &19.79 & 30.01 &50.48 & 81.21& 101.70\\
 \hline
 $h_3$& 31.58 & 45.50 & 73.07 & 114.27& 141.69\\ 
 $\mathbf{E} \left(\pi_3\right)$ &4.37 &6.94 &12.13 & 19.95& 25.18 \\
 \hline 
\end{tabular}

\medskip
\caption{This table presents the equilibrium outcomes for three players under the assumption that each player $i \in \left[ 3 \right]$ has a common ratio $\bar{f}_i/\bar{v}_i$, where the ratio is set to $2,3,5,8$, or $10$. The players' $\bar{v}_i$ follow the ratio $\bar{v}_1: \bar{v}_2:\bar{v}_3 = 5:2:1$, with $\bar{v}_3 = 30$.}
\label{table: 3 player with common ratio}
\end{table}

\begin{table}[h]
\centering
\begin{tabular}{ |c|c|c|c|c|c| } 
 \hline
 $\bar{v}_3$ & 10 & 20 & 30 & 50 & 80 \\ 
 \hline 
 $h_1$ & 6099.31 & 12198.62 &18297.94 & 30496.56 & 48794.50\\ 
 $\mathbf{E} \left(\pi_1\right)$ &37850.41 &75700.83 &113551.24 & 189252.07& 302803.32 \\
 \hline
 $h_2$& 932.00 & 1864.01 & 2796.01 & 4660.01 & 7456.02\\ 
 $\mathbf{E} \left(\pi_2\right)$ &140.91 &281.82 &422.74 & 704.56& 1127.30 \\
 \hline
 $h_3$& 49.89 & 99.79 & 149.68 & 249.47& 399.15\\ 
  $\mathbf{E} \left(\pi_3\right)$ &0.35 &0.71 &1.06 & 1.77& 2.83 \\
 \hline
\end{tabular}

\medskip
\caption{This table presents the equilibrium outcomes for three players under the assumption that $\frac{\bar{f}_1}{\bar{v}_1} = 1000 > \frac{\bar{f}_2}{\bar{v}_2} = 100 > \frac{\bar{f}_3}{\bar{v}_3} = 10 $. The players' $\bar{v}_i$ follow the ratio $\bar{v}_1: \bar{v}_2:\bar{v}_3 = 5:2:1$, where $\bar{v}_3$ is set to $10, 20, 30, 50$, or $80$.}
\label{table: 3 player with descending ratio}
\end{table}

\quad From Table~\ref{table: 3 player with common ratio}, we observe that, in most cases, the results align with the trend in Proposition~\ref{prop: 2 player NE ratios}: the more capable players pay relatively less. However, an exception occurs when $\bar{f}_i/\bar{v}_i = 2$ for $i=1,2,3$, $\bar{f}_2/\bar{f}_3 = \bar{v}_2/\bar{v}_3 = 2$ and $h_2/h_3 > 2$. This suggests that the presence of player $3$ may exert competitive pressure on player $2$, leading to a higher payment. Furthermore, the most capable player appears to gain a disproportionately higher expected utility, while the gap in expected utility between the second and least capable players narrows.

\quad From Table~\ref{table: 3 player with descending ratio}, we observe that when every player in the game is scaled by $k$ times, then their equilibrium solutions and expected utilities are also scaled by $k$ times. The effect observed in Table~\ref{table: 3 player with common ratio} is further amplified here. More capable players, such as player 1, attain disproportionately high expected utility, significantly exceeding $\frac{\bar{f}_1 \bar{v}_1}{\bar{f}_2 \bar{v}_2}$.

\subsection{Simulation of Validators' Stake Shares} \label{sec: simulation for validators' shares}

We simulate the evolution of total stake in the PoS system,
and the distribution of validators' stake shares. 
The parameters are set as follows: the number of validators is fixed at $N = 3$, with a time horizon of $T=1000$ steps. The initial stakes of the three validators are given by $s_{1,0} = 10, s_{2,0} = 20$, and $s_{3,0} = 30$, respectively. There are two builders, with bids set at $h_1 = 15$ and $h_2 = 20$. Additionally, we set $\mu = 0.7$, $\alpha = 8$, and $\beta_v = 11$, while $\gamma$ varies over the values $0, 0.1, 0.2,0.3$, and $1.5$. We run the simulation $1000$ times and plot the mean total stake and mean stake shares over time in Figure~\ref{fig: evolution of stakes}. Additionally, we repeat the simulation $10000$ times and plot the distribution of the final stake share, $\omega_{j,T}$, for $j = 1,2,3$, in Figure~\ref{fig: final stake share dist a>0}. We then set $\alpha = 0$ and plot the distribution of $\omega_{j,T}$ in Figure~\ref{fig: final stake share dist a=0}.


\begin{figure}[h]
    \centering
    \begin{minipage}{0.5\textwidth}
        \centering
        \includegraphics[width=\linewidth]{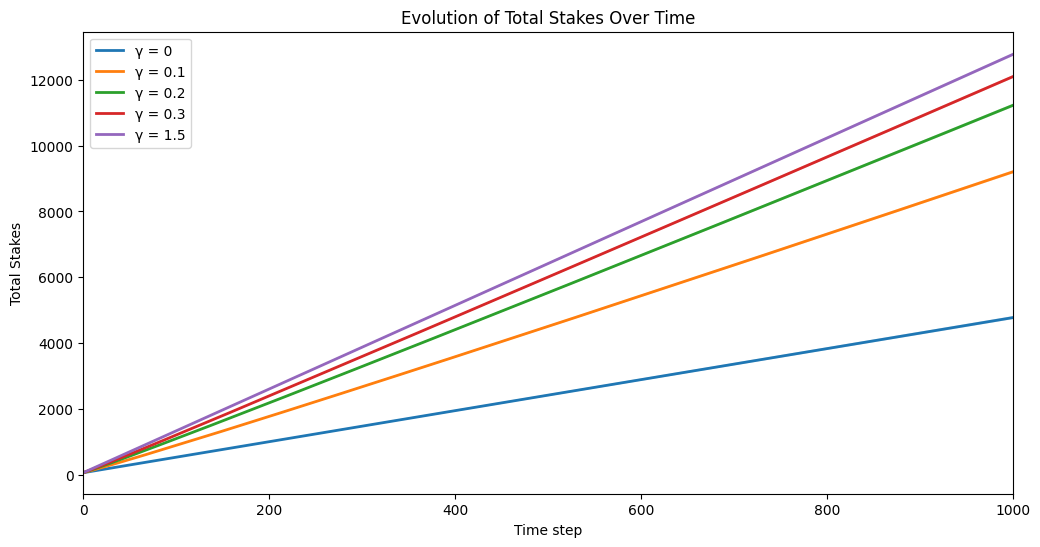}
        \small (a)
    \end{minipage}%
    \begin{minipage}{0.5\textwidth}
        \centering
        \includegraphics[width=\linewidth]{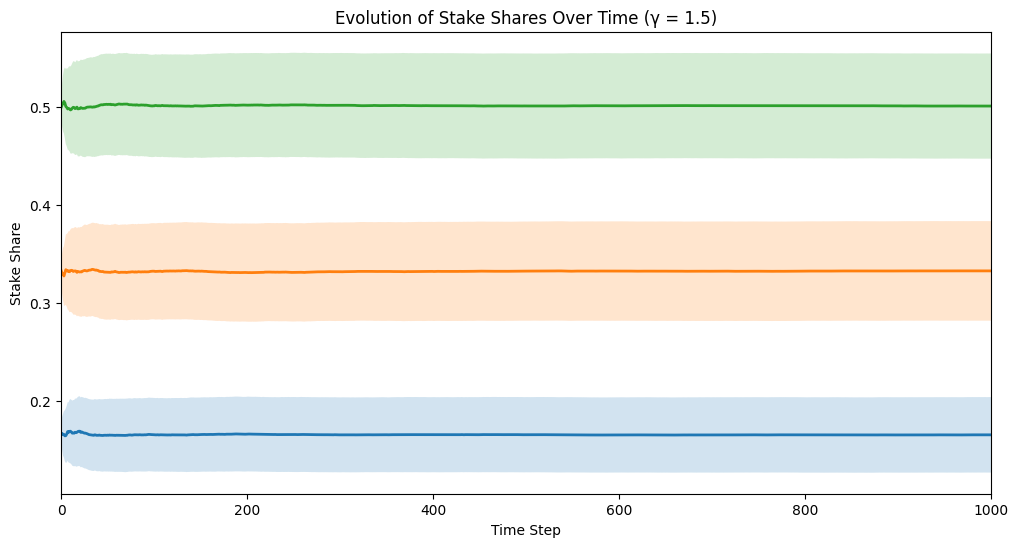}
        \small (b)
    \end{minipage}%
    
    \caption{\small (a) Evolution of mean total stakes over time when $\gamma = 0, 0.1, 0.2, 0.3$, and $1.5$. (b) Evolution of mean stake shares over time when $\gamma = 1.5$. The shaded region represents $\text{mean} \pm \frac{1}{4} \text{standard deviation}$.} 
    \label{fig: evolution of stakes}
\end{figure}

\begin{figure}[h]
    \centering  \includegraphics[width=\linewidth]{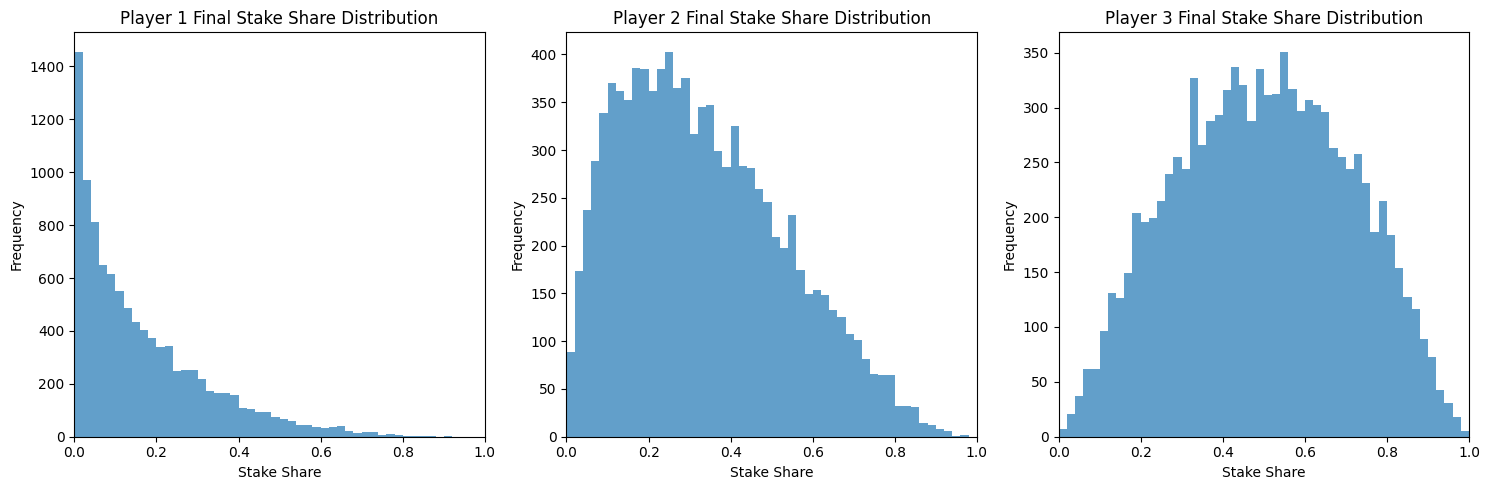}
    \caption{Distribution of final stake shares ($\alpha \neq 0, \gamma = 1.5$)} 
    \label{fig: final stake share dist a>0}
\end{figure}

\begin{figure}[h]
    \centering  \includegraphics[width=\linewidth]{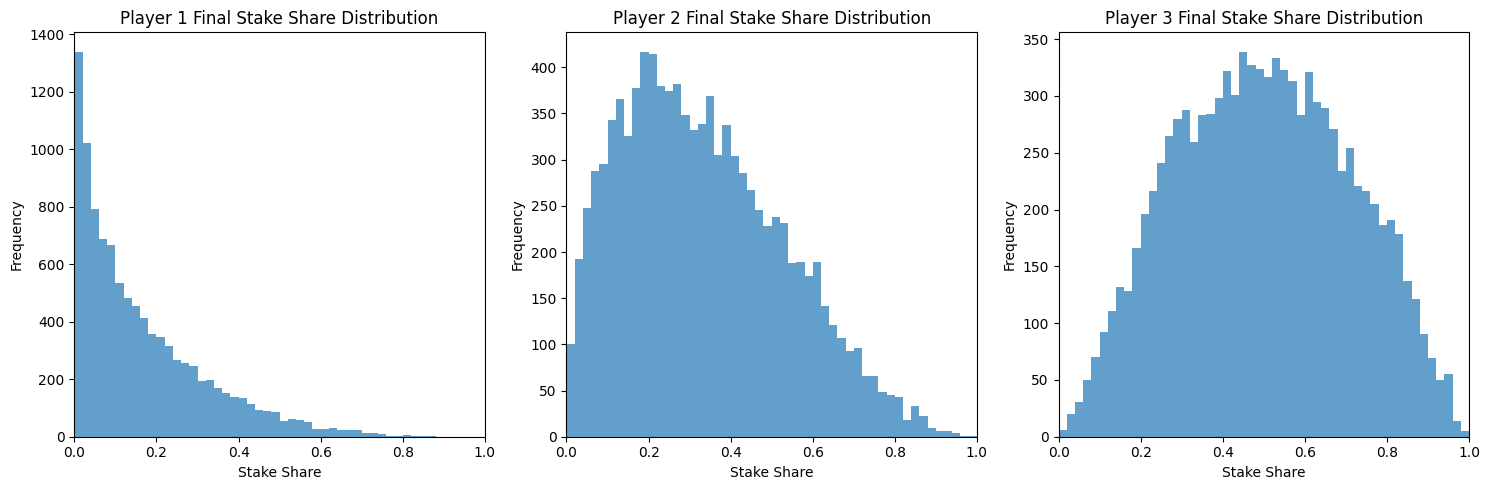}
    \caption{Distribution of final stake shares ($\alpha = 0, \gamma = 1.5$)} 
    \label{fig: final stake share dist a=0}
\end{figure}


\quad From Figure~\ref{fig: evolution of stakes}, we observe that the total stake grows at a rate of $R$ when $\gamma > 0$ and at a rate of $R - \alpha$ when $\gamma = 1.5$. The slight variations in growth rates for $\gamma = 0.1, 0.2, 0.3$ in the plot arise from the fact that $1/S_t^\gamma$ decays at different rates as $S_t \to \infty$. Extending the time horizon would further demonstrate that the growth rate converges to $R-\alpha$. 

\quad For the evolution of stake shares, Figure~\ref{fig: evolution of stakes} shows that they remain stable over time, consistent with their martingale property. Since the stake share dynamics exhibit the same behavior across different values of $\gamma$, we present only a representative case in the figure.

\quad Figure~\ref{fig: final stake share dist a>0} and Figure~\ref{fig: final stake share dist a=0} illustrate the marginal distribution of $\omega_{j,T}$ for $j=1,2,3$ in the cases where $\alpha > 0$ and $\alpha = 0$, respectively. Each player begins with initial stakes $10, 20$, and $30$, corresponding to initial stake shares of $\omega_{1,0} = \frac{1}{6}, \omega_{2,0} = \frac{1}{3}$, and $\omega_{3,0} = \frac{1}{2}$. The distributions in both cases appear highly similar, with their marginal distributions skewed toward their initial stake shares.

\quad When increasing the total initial stake while maintaining the initial relative shares, we observe that the limiting stake shares indeed converge to their initial values, aligning with the theoretical results established in Theorem~\ref{thm: stability of share}.

\quad However, when $\gamma=0$, we observe that under high consumption costs and small initial stakes, some validators' stake shares are likely to be driven to zero, as shown in Figures~\ref{fig: final stake share dist a=8 r=0} and~\ref{fig: 3d hist of final stake dist a = 8 r = 0} in Appendix~\ref{appendix: additional numerical results}. This phenomenon may be attributed to the fact that costs impose a greater relative burden on smaller validators, whereas larger validators experience only a limited impact. Specifically, since the term $\alpha s_{j,0}/S_0$ is bounded, larger validators are more resilient to costs, while smaller validators are more vulnerable and face greater challenges in accumulating stakes over time.

\section{Conclusions}
\label{concl}
\quad In this study, we examine the interactions between the order flow auction and the block-building auction, formulating the problem as a multiplayer game. We establish the existence of a Nash equilibrium in the general setting and prove its uniqueness in the two-player case, along with deriving the closed-form solution. Furthermore, we show that validators' stake shares follow a martingale process. We also conduct numerical simulations to complement our theoretical analysis. Both theoretical and numerical results suggest a tendency toward centralization in the builder space, which, in contrast, is highly unlikely in the validator space.


\bigskip
{\bf Acknowledgement}:
Ruofei Ma is supported by the Columbia Innovation Hub grant and the Center for Digital Finance and Technologies (CDFT) grant.
Wenpin Tang acknowledges financial support by NSF grant DMS-2206038, and the Tang Family Assistant Professorship.
The works of Ruofei Ma and David Yao are part of a Columbia-CityU/HK collaborative project that is supported by InnoHK Initiative, The Government of the HKSAR and the AIFT Lab.

\bibliographystyle{abbrvnat}
\bibliography{Reachablity}

@String{Academic = "Academic Press" }

@techreport{bains2022blockchain,
  author     = {Parma Bains},
  title      = {Blockchain Consensus Mechanisms: A Primer for Supervisors},
  institution = {International Monetary Fund},
  year       = {2022},
  number     = {}
}

@misc{ethereum_mev,
  author      = {Ethereum Foundation},
  title       = {Maximal Extractable Value ({M}{E}{V})},
  year        = {2024},  
  url         = {https://ethereum.org/en/developers/docs/mev/}
}

@misc{monoceros2023,
  author      = {Monoceros Ventures},
  title       = {Everything You Need to Know About Order Flow Auctions},
  year        = {2023},  
  url         = {https://www.monoceros.com/insights/order-flow-auctions}
}

@misc{bahrani2024,
      title={Centralization in Block Building and Proposer-Builder Separation}, 
      author={Maryam Bahrani and Pranav Garimidi and Tim Roughgarden},
      year={2024},
      eprint={2401.12120},
      archivePrefix={arXiv},
      primaryClass={cs.GT},
      url={https://arxiv.org/abs/2401.12120}, 
}

@misc{capponi2024,
  author      = {Agostino Capponi and Ruizhe Jia and Sveinn Olafsson},
  title       = {Proposer-Builder Separation, Payment for Order Flows, and Centralization in Blockchain},
  year        = {2024},
  note        = {Available at SSRN: https://ssrn.com/abstract=4723674},
  doi         = {10.2139/ssrn.4723674}
}

@misc{ethereum_pbs,
  author      = {Ethereum Foundation},
  title       = {Proposer-builder separation},
  year        = {2024},
  url         = {https://www.ethereum.org/en/roadmap/pbs/}
}

@misc{gupta2023,
      title={The Centralizing Effects of Private Order Flow on Proposer-Builder Separation}, 
      author={Tivas Gupta and Mallesh M Pai and Max Resnick},
      year={2023},
      eprint={2305.19150},
      archivePrefix={arXiv},
      primaryClass={econ.TH},
      url={https://arxiv.org/abs/2305.19150}, 
}

@misc{pai2023,
      title={Structural Advantages for Integrated Builders in {M}{E}{V}-{B}oost}, 
      author={Mallesh Pai and Max Resnick},
      year={2023},
      eprint={2311.09083},
      archivePrefix={arXiv},
      primaryClass={econ.TH},
      url={https://arxiv.org/abs/2311.09083}, 
}

@misc{yang2024,
      title={Decentralization of {E}thereum's Builder Market}, 
      author={Sen Yang and Kartik Nayak and Fan Zhang},
      year={2025},
      eprint={2405.01329},
      archivePrefix={arXiv},
      primaryClass={cs.CR},
      url={https://arxiv.org/abs/2405.01329}, 
}

@INPROCEEDINGS{daian2020,
  author={Daian, Philip and Goldfeder, Steven and Kell, Tyler and Li, Yunqi and Zhao, Xueyuan and Bentov, Iddo and Breidenbach, Lorenz and Juels, Ari},
  booktitle={2020 IEEE Symposium on Security and Privacy (SP)}, 
  title={Flash Boys 2.0: Frontrunning in Decentralized Exchanges, Miner Extractable Value, and Consensus Instability}, 
  year={2020},
  volume={},
  number={},
  pages={910-927},
  doi={10.1109/SP40000.2020.00040}}

@misc{buterin2021,
  author      = {Vitalik Buterin}, 
  title       = {Proposer/block builder separation-friendly fee market designs},
  year        = {2021},  
  url         ={https://ethresear.ch/t/proposer-block-builder-separation-friendly-fee-market-designs/9725}
}

@misc{flashbots2024,
  author      = {Flashbots},
  title       = {Flashbots {M}{E}{V}-{B}oost: Introduction},
  year        = {2024},  
  url         = {https://docs.flashbots.net/flashbots-mev-boost/introduction},
  note        = {Flashbots Documentation}
}

@misc{gosselin2023,
  author      = {Stephane Gosselin and Ankit Chiplunkar},
  title       = {The Orderflow Auction Design Space},
  year        = {2023}, 
  url         = {https://frontier.tech/the-orderflow-auction-design-space}
}

@misc{wang2024privateorderflowsbuilder,
      title={Private Order Flows and Builder Bidding Dynamics: The Road to Monopoly in {E}thereum's Block Building Market}, 
      author={Shuzheng Wang and Yue Huang and Wenqin Zhang and Yuming Huang and Xuechao Wang and Jing Tang},
      year={2024},
      eprint={2410.12352},
      archivePrefix={arXiv},
      primaryClass={cs.CE},
      url={https://arxiv.org/abs/2410.12352}, 
}

@misc{titan2023builder,
  author       = {Titan},
  title        = {Builder Dominance and Searcher Dependence},
  journal      = {Frontier Research},
  year         = {2023},
  url          = {https://frontier.tech/builder-dominance-and-searcher-dependence}
}

@article{tang2024polynomialvotingrules,
  title={Polynomial voting rules},
  author={Tang, Wenpin and Yao, David D},
  journal={Mathematics of Operations Research},
  year={2024},
  volume={50},
  number = {1},
  pages = {90-106},
}

@inproceedings{Tang24,
author = {Wenpin Tang },
title = {Trading and wealth evolution in the Proof of Stake protocol},
booktitle = {Proof-of-Stake for Blockchain Networks},
year = {2024},
chapter = {Chapter 7},
pages = {135-161},
doi = {10.1049/PBSE024E_ch7},
URL = {https://digital-library.theiet.org/doi/abs/10.1049/PBSE024E_ch7},
}

@article{rosu2021market,
  author    = {Ioanid Rosu and Fahad Saleh},
  title     = {Evolution of Shares in a Proof-of-Stake Cryptocurrency},
  journal   = {Management Science},
  year      = {2021},
  volume    = {67},  
  pages     = {661-672}, 
  url       = {http://dx.doi.org/10.2139/ssrn.3377136} 
}

@misc{Tang22,
  title={Stability of shares in the {P}roof of {S}take protocol -- concentration and phase transitions},
  author={Tang, Wenpin},
  year={2022},
  eprint={2206.02227},
  archivePrefix={arXiv},
  primaryClass={econ.GN},
  url={https://arxiv.org/abs/2206.02227},
}

@misc{ethresear2024buildersbehavioralprofiles,
  author       = {Thomas Thiery},
  title        = {Empirical Analysis of Builders' Behavioral Profiles (BBPs)},
  howpublished = {\url{https://ethresear.ch/t/empirical-analysis-of-builders-behavioral-profiles-bbps/16327/3}},
  year         = {2024},
  note         = {Accessed: 2024-11-16}
}

@article{giacomo2007urn,
 URL = {http://www.jstor.org/stable/20443536},
 author = {Giacomo Aletti and Caterina May and Piercesare Secchi},
 journal = {Advances in Applied Probability},
 number = {3},
 pages = {690--707},
 publisher = {Applied Probability Trust},
 title = {On the Distribution of the Limit Proportion for a Two-Color, Randomly Reinforced Urn with Equal Reinforcement Distributions},
 urldate = {2025-02-10},
 volume = {39},
 year = {2007}
}

@Article{wu2018energy,
AUTHOR = {Wu, Jiani and Tran, Nguyen Khoi},
TITLE = {Application of Blockchain Technology in Sustainable Energy Systems: An Overview},
JOURNAL = {Sustainability},
VOLUME = {10},
YEAR = {2018},
NUMBER = {9},
ARTICLE-NUMBER = {3067},
URL = {https://www.mdpi.com/2071-1050/10/9/3067},
DOI = {10.3390/su10093067}
}

@article{ratta2021healthcare,
author = {Ratta, Pranav and Kaur, Amanpreet and Sharma, Sparsh and Shabaz, Mohammad and Dhiman, Gaurav},
title = {Application of Blockchain and Internet of Things in Healthcare and Medical Sector: Applications, Challenges, and Future Perspectives},
journal = {Journal of Food Quality},
volume = {2021},
number = {1},
pages = {7608296},
doi = {10.1155/2021/7608296},
url = {https://onlinelibrary.wiley.com/doi/abs/10.1155/2021/7608296},
eprint = {https://onlinelibrary.wiley.com/doi/pdf/10.1155/2021/7608296},
year = {2021}
}

@article{hashemi2020cryptocurrency,
  author    = {Mohammad Hashemi Joo and Yuka Nishikawa and Krishnan Dandapani},
  title     = {Cryptocurrency, a successful application of blockchain technology},
  journal   = {Managerial Finance},
  year      = {2020},
  volume    = {46},
  number    = {6},
  pages     = {715--733},
  doi       = {10.1108/MF-09-2018-0451},
  url       = {https://doi.org/10.1108/MF-09-2018-0451}
}

@incollection{LAROIYA2020213,
title = {Chapter 9 - Applications of Blockchain Technology},
editor = {Saravanan Krishnan and Valentina E. Balas and E. Golden Julie and Y. Harold Robinson and S. Balaji and Raghvendra Kumar},
booktitle = {Handbook of Research on Blockchain Technology},
publisher = {Academic Press},
pages = {213-243},
year = {2020},
isbn = {978-0-12-819816-2},
doi = {10.1016/B978-0-12-819816-2.00009-5},
url = {https://www.sciencedirect.com/science/article/pii/B9780128198162000095},
author = {Chetna Laroiya and Deepika Saxena and C. Komalavalli}
}

@Article{cryptography3010003,
AUTHOR = {Siyal, Asad Ali and Junejo, Aisha Zahid and Zawish, Muhammad and Ahmed, Kainat and Khalil, Aiman and Soursou, Georgia},
TITLE = {Applications of Blockchain Technology in Medicine and Healthcare: Challenges and Future Perspectives},
JOURNAL = {Cryptography},
VOLUME = {3},
YEAR = {2019},
NUMBER = {1},
ARTICLE-NUMBER = {3},
URL = {https://www.mdpi.com/2410-387X/3/1/3},
DOI = {10.3390/cryptography3010003}
}

@article{bell2018applications,
  author    = {Liam Bell and William J. Buchanan and Jonathan Cameron and Owen Lo},
  title     = {Applications of Blockchain Within Healthcare},
  journal   = {Blockchain in Healthcare Today},
  year      = {2018},
  doi = {10.30953/bhty.v1.8},
  url       = {https://doi.org/10.30953/bhty.v1.8}
}

@article{KIM20202561,
title = {A Study on the Application of Blockchain Technology in the Construction Industry},
journal = {KSCE Journal of Civil Engineering},
volume = {24},
number = {9},
pages = {2561-2571},
year = {2020},
doi = {10.1007/s12205-020-0188-x},
url = {https://www.sciencedirect.com/science/article/pii/S1226798824044593},
author = {Kyeongbaek Kim and Gayeoun Lee and Sangbum Kim}
}

@article{8344bf6b-b18f-3bf0-88f4-c6a0b490d844,
 ISSN = {00129682, 14680262},
 URL = {http://www.jstor.org/stable/1911749},
 author = {J. B. Rosen},
 journal = {Econometrica},
 number = {3},
 pages = {520--534},
 publisher = {[Wiley, Econometric Society]},
 title = {Existence and Uniqueness of Equilibrium Points for Concave N-Person Games},
 urldate = {2025-07-29},
 volume = {33},
 year = {1965}
}

\newpage

\appendix 

\begin{center}
 \large \textbf{Appendix} 
\end{center}

\section{Code for Problem~\ref{optimization problem (solutions)}}  \label{appendix: codes}
\lstset{
    language=Python,           basicstyle=\ttfamily\footnotesize,
    keywordstyle=\color{blue},
    commentstyle=\color{green!50!black},
    stringstyle=\color{red},
    breaklines=true,
    numbers=left,
    numberstyle=\tiny,
    frame=single,
    captionpos=b
}

\begin{lstlisting}[caption={}, label={}]
import numpy as np
from scipy.optimize import minimize

def func(vars):
    f1, f2, v1, v2 = vars
    try:
        p = -((11*f1**2 + 12*v1**2 + 4*f1*(4*f2 + v1 + 8*v2))/(8*f1**2))
        q = (3*f1**3 + 8*v1**3 + 4*f1*v1*(4*f2 + v1 + 8*v2) + f1**2*(-10*v1 + 32*v2))/(8*f1**3)
        delta0 = (2*f1-2*f2+4*v1-4*v2)**2+3*(3*f1+2*v1)*(3*f2+2*v2)-12*f1*f2
        delta1 = -27*f2*(3*f1 + 2*v1)**2 + 72*f1*f2*(2*f1 - 2*f2 + 4*v1 - 4*v2) + 2*(2*f1 - 2*f2 + 4*v1 - 4*v2)**3 - \
                  9*(3*f1 + 2*v1)*(2*f1 - 2*f2 + 4*v1 - 4*v2)*(-3*f2 - 2*v2) + 27*f1*(-3*f2 - 2*v2)**2

        phi = np.arccos(delta1/(2*np.sqrt(delta0**3)))

        S = 1/2*np.sqrt(-2*p/3 + 2/(3*f1)*np.sqrt(delta0)*np.cos(phi/3))

        lmda = -(3*f1+2*v1)/(4*f1)+S+1/2*np.sqrt(-4*S**2-2*p-q/S)
        h1_1 = lmda*(f1*lmda+f1+2*v1)/((1+2*lmda)*(1+lmda))
        h1_2 = (f2 + (f2+2*v2)*lmda)/(lmda*(2+lmda)*(1+lmda))

        return (8*S**3-q)
    except:
        return np.inf  

initial_guess = [1, 1, 1, 1]  

# Variable bounds (all variables > 0)
bounds = [(1e-5, None), (1e-5, None), (1e-5, None), (1e-5, None)]  

# Define the constraints f1 >= v1 and f2 >= v2
constraints = [
    {"type": "ineq", "fun": lambda vars: vars[0] - vars[2]},  # f1 - v1 >= 0
    {"type": "ineq", "fun": lambda vars: vars[1] - vars[3]},  # f2 - v2 >= 0
]

result = minimize(func, initial_guess, method='SLSQP', bounds=bounds, constraints=constraints)

# Results
if result.success:
    print("Global minimum found:")
    print("Function value:", result.fun)
    print("At variables (f1, f2, v1, v2):", result.x)
else:
    print("Optimization failed:", result.message)

\end{lstlisting}

\section{Proof of Theorem~\ref{thm: stability of share}} \label{proof: thm-stability}
\quad To prove Theorem~\ref{thm: stability of share}, we need a series of lemmas.

\begin{lemma} \label{lemma: conditional var of share}
    Let $\text{Var}_t \left( \,\cdot\, \right)$ and $\mathbf{E}_t \left( \,\cdot\, \right) $ denote the conditional variance and conditional expectation at time $t$, respectively, i.e., $\text{Var}_t \left( \,\cdot\, \right) = \text{Var}\left( \,\cdot\, \right \rvert \mathcal{F}_t)$, $\mathbf{E}_t \left( \,\cdot\, \right) = \mathbf{E}\left( \,\cdot\, \right \rvert \mathcal{F}_t)$. When $\gamma = 0$, the conditional variance at time $t$ of validator $j$'s share at time $t+1$ is:
    \begin{equation*}
        \text{Var}_t \left( \omega_{j, t+1} \right) =  \omega_{j,t} \left( 1 - \omega_{j,t}\right) \mathbf{E}_t \left[\left( \frac{R_{t+1}}{S_{t+1}}\right)^2 \right].
    \end{equation*}
\end{lemma}

\begin{proof}
When $\gamma = 0$,
\begin{equation*}
    \omega_{j,t+1} = \omega_{j,t} \frac{S_t - \alpha}{S_{t+1}} + \frac{R_{t+1}}{S_{t+1}} \mathbf{1} \left( X_{j,t+1}\right).
\end{equation*}
Noting that $\frac{S_t - \alpha}{S_{t+1}} = 1-\frac{R_{t+1}}{S_{t+1}}$, we have
\begin{align*}
    \mathbf{E}_t \left( \omega_{j, t+1}^2 \right) &= \mathbf{E}_t \left[ \omega_{j,t}^2 \left(1-\frac{R_{t+1}}{S_{t+1}}\right)^2 + \left(\frac{R_{t+1}}{S_{t+1}}\right)^2 \mathbf{1} \left( X_{j,t+1}\right) + 2 \omega_{j,t} \mathbf{1} \left( X_{j,t+1}\right) \frac{R_{t+1}}{S_{t+1}} \left(1-\frac{R_{t+1}}{S_{t+1}}\right)
    \right],\\ &=\omega_{j,t}^2 + \omega_{j,t} \left( 1-\omega_{j,t} \right) \mathbf{E}_t \left[ \left( \frac{R_{t+1}}{S_{t+1}} \right)^2 \right].
\end{align*}
where the last equality is obtained by Assumption~\ref{assump: independence of validator selection}.
By Theorem~\ref{thm: martingale--validator}, $\left(\omega_{j,t}\right)_{t \geq 0}$ is a martingale:
\begin{equation*}
    \mathbf{E}_t \left( \omega_{j,t+1} \right) = \omega_{j,t}.
\end{equation*}
Therefore, the conditional variance of $\omega_{j,t+1}$ is given by
\begin{align*}
    \text{Var}_t \left( \omega_{j, t+1} \right) &= \mathbf{E}_t \left( \omega_{j, t+1}^2 \right)-\left(\mathbf{E}_t \left( \omega_{j,t+1} \right)\right)^2, \\
    &=\omega_{j,t} \left( 1 - \omega_{j,t}\right) \mathbf{E}_t \left[\left( \frac{R_{t+1}}{S_{t+1}}\right)^2 \right].
\end{align*}
\end{proof}

\begin{lemma} \label{lemma: unconditional var of share}
    When $\gamma = 0$, the unconditional variance of validator $j$'s share at time $t$ is given by
    \begin{equation}
        \text{Var} \left( \omega_{j,t} \right) = a_t \omega_{j,0} \left( 1-\omega_{j,0}\right), \label{eq: unconditional var of share}
    \end{equation}
where the sequence $a_t$ satisfies
\begin{equation}
    a_1 = z_1, \quad a_{t+1} = a_t + z_{t+1} \left( 1- a_t \right), \label{eq: a_t}
\end{equation}
where
\begin{equation}
    z_{t+1} = \mathbf{E} \left[\left( \frac{R_{t+1}}{S_{t+1}}\right)^2\right]. \label{eq: z_t}
\end{equation}
\end{lemma}

\begin{proof}
We prove by induction, following the approach of~\citet{rosu2021market}. Since $\mathbf{E}_t \left( \omega_{j,t+1} \right) = \omega_{j,t}$, the following equation holds for all $t \geq 0$:
    \begin{equation}
        \text{Var} \left( \omega_{j,t+1} \right) = \text{Var} \left( \omega_{j,t} \right) + \mathbf{E} \left(\text{Var}_t \left( \omega_{j, t+1} \right) \right). \label{eq: general formula for unconditional var of share}
    \end{equation}
From Lemma~\ref{lemma: conditional var of share}, we establish the base case at $t=1$:
\begin{align*}
    \text{Var} \left( \omega_{j,1} \right) &= \omega_{j,0} \left( 1- \omega_{j,0}\right) \mathbf{E} \left[\left( \frac{R_{1}}{S_{1}}\right)^2 \right], \\
    &= z_1 \omega_{j,0} \left( 1- \omega_{j,0}\right).
\end{align*}
Next, we assume that Eq.~(\ref{eq: unconditional var of share}) holds for all time steps up to $t = k$, where $k > 1$. At time $t=k+1$, we have 
\begin{align*}
    \text{Var} \left( \omega_{j,t+1} \right) &= \text{Var} \left( \omega_{j,t} \right) + \mathbf{E} \left\{ \omega_{j,t} \left( 1 - \omega_{j,t}\right) \mathbf{E}_t \left[\left( \frac{R_{t+1}}{S_{t+1}}\right)^2 \right] \right\},\\
\end{align*}
from Lemma~\ref{lemma: conditional var of share} and Eq.~(\ref{eq: general formula for unconditional var of share}). Since $\text{Var} \left( \omega_{j,t} \right) = a_t \omega_{j,0} \left( 1-\omega_{j,0}\right)$, we have
\begin{align*}
    \text{Var} \left( \omega_{j,t+1} \right) &= a_t \omega_{j,0} \left( 1 - \omega_{j,0} \right) + \left( \omega_{j,0} -a_t \omega_{j,0} \left( 1-\omega_{j,0} \right) - \omega_{j,0}^2 \right) \cdot \mathbf{E} \left[\left( \frac{R_{t+1}}{S_{t+1}}\right)^2 \right], \\
    &=\omega_{j,0} \left( 1-\omega_{j,0} \right) \left\{ a_t + \left( 1- a_t \right) \mathbf{E} \left[\left( \frac{R_{t+1}}{S_{t+1}}\right)^2\right] 
    \right\}.
\end{align*}
This completes the induction step and proves the lemma.
\end{proof}

\begin{lemma} \label{lemma: a_t in 0,1 and increasing}
    When $\gamma = 0$, let the sequence $\left\{ a_t\right\}_{t\geq 0}$ be defined as in Eqs.~(\ref{eq: a_t})–(\ref{eq: z_t}) and extend it by setting $a_0 = 0$. This sequence satisfies $a_t \in \left[ 0,1\right]$ and is non-decreasing, i.e., $a_{t+1} \ge a_t$, for all $t$.
\end{lemma}

\begin{proof}
We prove by induction. When $\gamma = 0$, the sequence $\left\{ S_t \right\}$ satisfies the recurrence relation
\begin{equation*}
    S_{t+1} = S_t + R_{t+1} - \alpha,
\end{equation*}
as given by Eq.~(\ref{eq:S_t}).
For the base case $t=1$, by Eqs.~(\ref{eq: a_t})–(\ref{eq: z_t}), we have 
\begin{equation*}
    a_1 = \mathbf{E} \left[ \left(  \frac{R_1}{S_1}\right)^2 \right]= \mathbf{E} \left[ \left(  \frac{R_1}{S_0 + R_1 - \alpha}\right)^2 \right] \in \left[0,1 \right],
\end{equation*}
since $\alpha < S_0$ by Assumption~\ref{assump: bound on alpha}. Clearly, $a_1 \geq a_0 = 0$. Suppose that $a_t \in \left[ 0,1\right]$ and $a_{t+1} \ge a_t$, for all $1\leq t\leq k$. At $t=k+1$, we have $a_{k+1} = a_k + z_{k+1} \left( 1-a_k\right)$. By the definition of $z_t$, \begin{align*}
    z_{k+1} &= \mathbf{E} \left[\left( \frac{R_{k+1}}{S_{k+1}}\right)^2\right],\\
    &=\mathbf{E} \left[\left( \frac{R_{k+1}}{S_0 + \sum_{i=1}^{k+1}R_i - \left(k+1 \right) \alpha}\right)^2\right] \in \left[0,1\right], \label{eq: z_t}
\end{align*}
since $\alpha < S_0$ and $\alpha < R_{\min}$ by Assumption~\ref{assump: bound on alpha}. Therefore, it follows that $a_{k+1} \geq a_k$ and $a_{k+1} \in \left[ 0,1 \right]$. This completes the induction step and the lemma is proved.
\end{proof}

\begin{lemma}
    Let $a_t$ be defined by Eqs.~(\ref{eq: a_t})–(\ref{eq: z_t}). We have for each $t \geq 1$,
\begin{equation}
    a_t \le \frac{R_{\max}^2}{S_0 \left( R_{\min} - \alpha \right)}. \label{eq: upper bound for at}
\end{equation}
\end{lemma}

\begin{proof}
    By Eqs.~(\ref{eq: a_t})–(\ref{eq: z_t}), we have
    \begin{equation*}
        a_t - a_0 = \sum_{i=1}^t z_i \left( 1-a_{i-1}\right),
    \end{equation*}
noting that $a_0 \coloneqq 0$. It follows from Lemma~\ref{lemma: a_t in 0,1 and increasing} that $a_t \in \left[0,1\right]$. Consequently, we obtain
\begin{align*}
    a_t &\leq \sum_{i=1}^t z_i=\sum_{i=1}^t \mathbf{E} \left[\left( \frac{R_{i}}{S_{i}}\right)^2\right], \\
    & \leq \sum_{i=1}^t \left(\frac{R_{\max}}{S_0 + i\left(  R_{\min} - \alpha\right)}\right)^2,
\end{align*}
where the last inequality follows from the fact that $S_i = S_0 + \sum_{j=1}^i R_i - i \alpha$.
Since the function $x \to \left(\frac{R_{\max}}{S_0 + x \left(  R_{\min} - \alpha\right)}\right)^2$ is decreasing on $\mathbb{R}_+$ due to $R_{\max} >0, S_0 >0$, and $R_{\min} - \alpha > 0$ (by Assumption~\ref{assump: bound on alpha}), using the sum-integral trick as in ~\citet[Lemma A.4]{rosu2021market}, we obtain:
\begin{equation}
    a_t \leq \int_{0}^\infty \left(\frac{R_{\max}}{S_0 + x\left(  R_{\min} - \alpha\right)}\right)^2 dx = \frac{R_{\max}^2}{S_0 \left( R_{\min} - \alpha \right)}.
    \label{ineq: upper bound on at}
\end{equation}
\end{proof}

\textit{Proof of Theorem~\ref{thm: stability of share}.}
By Lemma~\ref{lemma: unconditional var of share}, Chebyshev's inequality, and the uppder bound in~\ref{eq: upper bound for at}, we get
    \begin{equation*}
        \mathbf{P} \left( \left\lvert \frac{\omega_{j,t}}{\omega_{j,0}} - 1\right\lvert > \epsilon \right) \leq \frac{a_t \left( 1- \omega_{j,0}\right) }{\epsilon^2 \omega_{j,0}} \leq \frac{R_{\max}^2}{\epsilon^2 s_{j,0} \left( R_{\min} - \alpha \right)},
    \end{equation*}
 since $S_0 \omega_{j,0} = s_{j,0}$ and $0 \leq 1-\omega_{j,0} \leq 1$. This proves~the estimate \eqref{ineq: limiting share cvg to initial share}.

\newpage
\section{Additional Results for Numerical Experiments} \label{appendix: additional numerical results}

\begin{table}[h] 
\centering
\begin{tabular}{ |c|c|c|c|c|c| } 
 \hline
 \rule{0pt}{10pt} $\bar{f}_i/\bar{v}_i$ & 2 & 3 & 5 & 8 & 10 \\ 
 \hline 
 $h_1$ & 154.60 & 209.39 &319.68 & 485.63 &596.37 \\ 
 $\mathbf{E} \left(\pi_1\right)$ &94.89 &133.59 &211.51 &328.74 & 406.96 \\
 \hline
 $h_2$& 96.45 & 134.13 & 209.45 & 322.38 & 397.67 \\ 
 $\mathbf{E} \left(\pi_2\right)$ &30.28 & 45.02 &74.48 &118.65 & 148.09\\
 \hline
 $h_3$& 63.84 & 90.77 & 144.32 &224.47 & 277.87\\ 
 $\mathbf{E} \left(\pi_3\right)$ &12.40 &19.03 &32.30 &52.21 &65.49 \\
 \hline 
 $h_4$& 31.16 &45.46 &73.91 &116.47 &144.82\\
 $\mathbf{E} \left(\pi_4\right)$ &2.84 &4.49 &7.82& 12.83& 16.18 \\
 \hline
\end{tabular}

\medskip
\caption{This table presents the equilibrium outcomes for four players under the assumption that each player $i \in \left[ 4 \right]$ has a common ratio $\bar{f}_i/\bar{v}_i$, where the ratio is set to $2,3,5,8$, or $10$. The players' $\bar{v}_i$ follow the ratio $\bar{v}_1: \bar{v}_2:\bar{v}_3:\bar{v}_4 = 5:3:2:1$, with $\bar{v}_4 = 30$.}
\label{table: 4 player with common ratio}
\end{table}

\begin{table}[h]
\centering
\begin{tabular}{ |c|c|c|c|c|c| } 
 \hline
 $\bar{v}_4$ & 10 & 20   & 30  & 50  & 80  \\ 
 \hline 
 $h_1$ & 13099.66  &   26199.33&39298.99  & 65498.31  & 104797.30 \\ 
 $\mathbf{E} \left(\pi_1\right)$ & 23844.40  & 47688.79 &71533.19  & 119221.98 & 190755.16  \\
 \hline
 $h_2$&6164.17   & 12328.35  &18492.52  & 30820.87  &49313.39 \\ 
 $\mathbf{E} \left(\pi_2\right)$ &2687.11  &5374.23 &8061.34 &13435.56 &21496.90  \\
 \hline
 $h_3$&976.29  &1952.57  & 2928.86 &4881.44 &7810.30 \\ 
  $\mathbf{E} \left(\pi_3\right)$ &49.30 & 98.61&147.91 &246.52 &394.43  \\
 \hline
 $h_4$& 49.96 &99.93 &149.89 &249.81 &399.70\\
 $\mathbf{E} \left(\pi_4\right)$ &0.12 &0.25 &0.37 & 0.62& 0.99\\
 \hline
\end{tabular}

\medskip
\caption{This table presents the equilibrium outcomes for four players under the assumption that $\frac{\bar{f}_1}{\bar{v}_1} = 1000 > \frac{\bar{f}_2}{\bar{v}_2} = 500 > \frac{\bar{f}_3}{\bar{v}_3} = 100 > \frac{\bar{f}_4}{\bar{v}_4} = 10$. The players' $\bar{v}_i$ follow the ratio $\bar{v}_1: \bar{v}_2:\bar{v}_3: \bar{v}_4= 5:3:2:1$, where $\bar{v}_4$ is set to $10, 20, 30, 50$, or $80$.}
\label{table: 4 player with descending ratio}
\end{table}

\begin{figure}[h]
    \centering  \includegraphics[width=\linewidth]{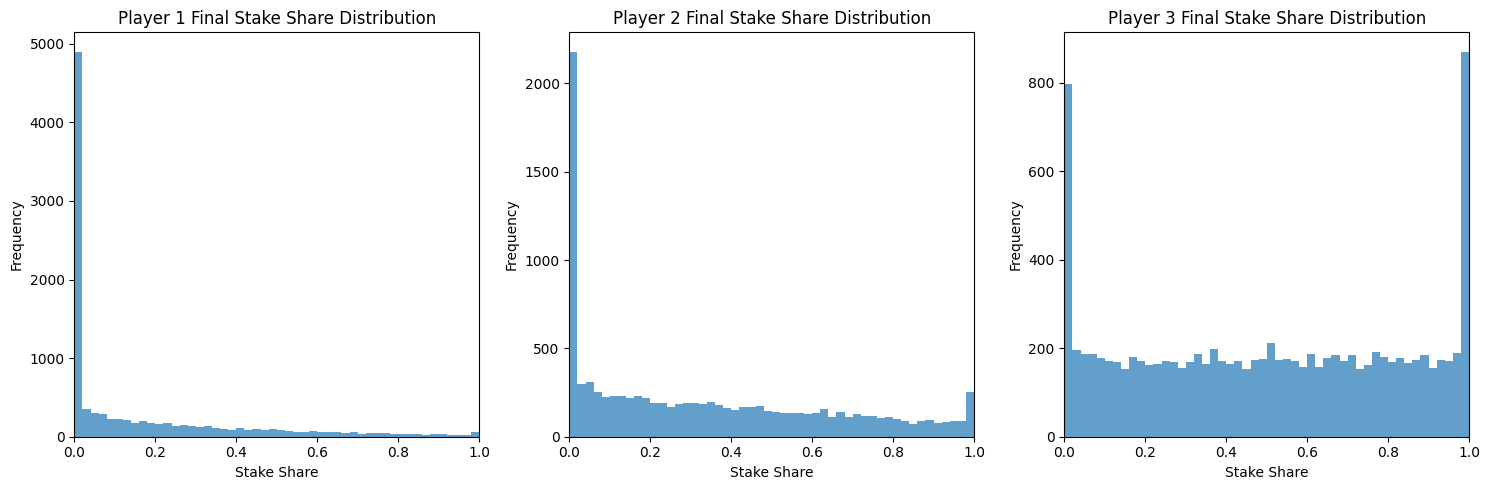}
    \caption{Distribution of final stake shares ($\alpha \neq 0, \gamma = 0$)} 
    \label{fig: final stake share dist a=8 r=0}
\end{figure}

\begin{figure}[h]
    \centering  \includegraphics[scale=0.2]{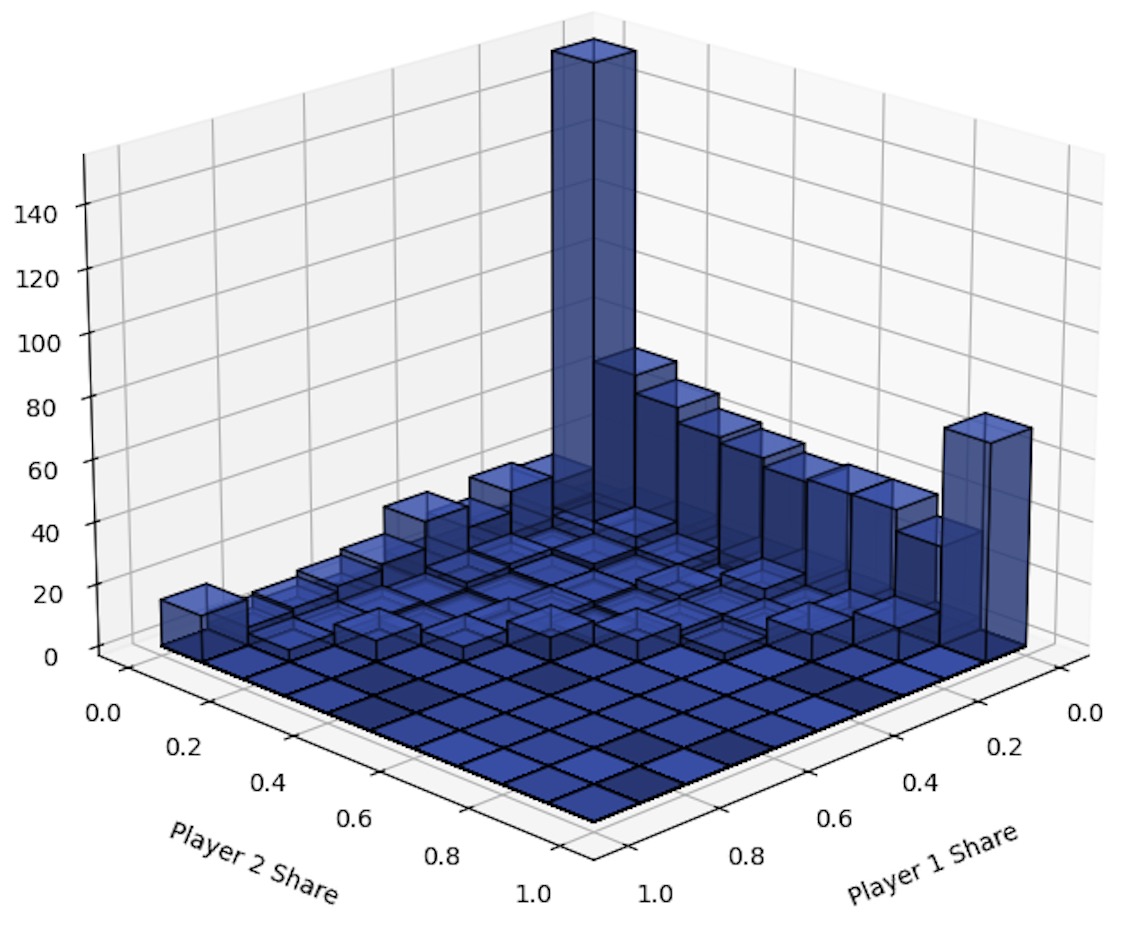}
    \caption{Histogram of the joint distribution of final stake shares ($\alpha \neq 0, \gamma = 0$)} 
    \label{fig: 3d hist of final stake dist a = 8 r = 0}
\end{figure}

\begin{figure}[h]
    \centering  \includegraphics[width=\linewidth]{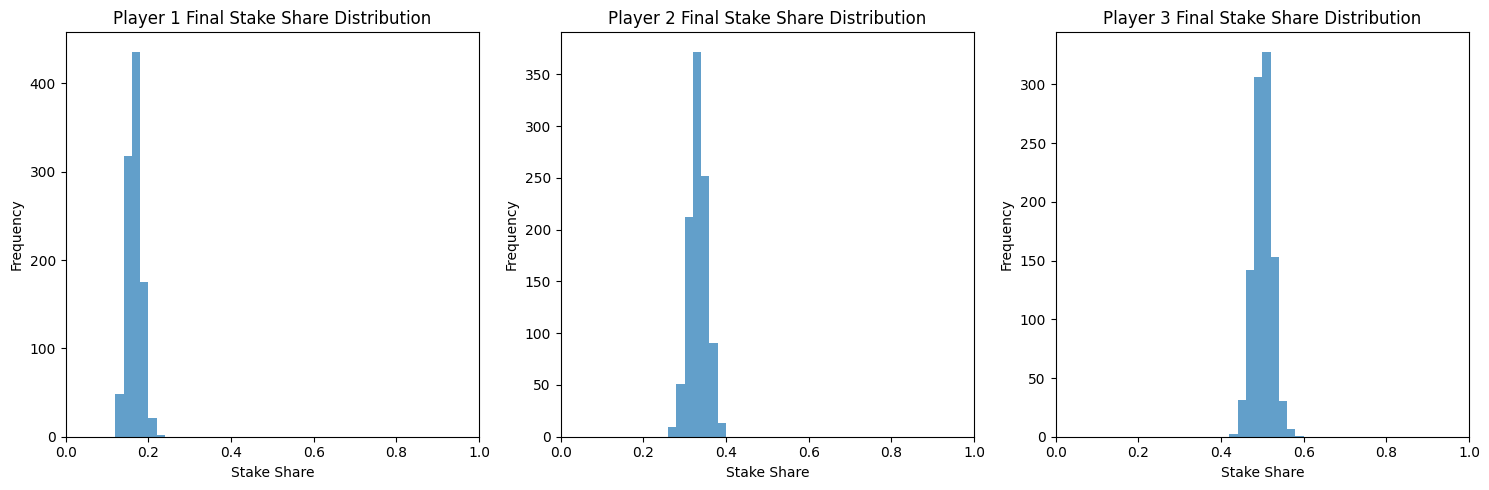}
    \caption{Distribution of final stake shares with large initial stakes ($\alpha \neq 0, \gamma = 1.5$)} 
    \label{fig: final stake share dist a=8 r=1.5 initial stake large}
\end{figure}

\quad Figure~\ref{fig: final stake share dist a=8 r=0} illustrates the limiting stake distribution under conditions of high consumption costs and small initial stakes, specifically with $\gamma = 0, \alpha = 8$, and initial stakes $s_{1,0} = 10, s_{2,0} = 20, s_{3,0} = 30$. The reward $R_t$ follows a similar setting as in previous cases. Under this setting, we observe that some validators' stake shares are likely to be driven toward zero, as explained in Section~\ref{sec: simulation for validators' shares}. Their joint distribution is shown in Figure~\ref{fig: 3d hist of final stake dist a = 8 r = 0}.

\quad Compared to the previous setting, where the initial stakes for the three players are $10, 20$, and $30$, as shown in Figures~\ref{fig: final stake share dist a>0},~\ref{fig: final stake share dist a=0}, and~\ref{fig: final stake share dist a=8 r=0}, we set the initial stakes to $1000$, $2000$, and $3000$, respectively, in Figure~\ref{fig: final stake share dist a=8 r=1.5 initial stake large}. In this setting, we observe that the limiting share, $\omega_{j,T}$, indeed converges to the initial share, consistent with the theoretical results established in Theorem~\ref{thm: stability of share}.

\end{document}